\documentclass[11pt, letterpaper]{article}

\usepackage{microtype}
\usepackage{mathtools,mathrsfs}
\usepackage[table,svgnames]{xcolor}
\usepackage[colorlinks=true,linkcolor=black,citecolor=MidnightBlue,urlcolor=MidnightBlue]{hyperref}
\usepackage{xspace}
\usepackage{tikz}
\usepackage{fullpage}

\usepackage{amssymb,amsmath,amsfonts,fullpage,mdwlist,enumerate,amsthm, graphicx,algpseudocode, algorithm, algorithmicx, cite}
\usepackage[letterpaper, margin=1in]{geometry}
\usepackage{thmtools, thm-restate}
\usepackage{enumitem}

\newtheorem{theorem}{Theorem}[section]
\newtheorem{lemma}[theorem]{Lemma}
\newtheorem{observation}{Observation}

\newtheorem{claim}[theorem]{Claim}

\title{Distributed Spanner Approximation}

\author{Keren Censor-Hillel\footnote{Technion, Department of Computer Science, \texttt{\{ckeren,smichald\}@cs.technion.ac.il}. Supported in part by the Israel Science Foundation (grant 1696/14).}
\and Michal Dory\footnotemark[1]{}
}

\newcommand{\remove}[1]{}

\newcommand{\local}{\textsc{Local}\xspace}
\newcommand{\congest}{\textsc{Congest}\xspace}

\begin{document}

\begin{titlepage}

\maketitle

\begin{abstract}
We address the fundamental network design problem of constructing
approximate minimum spanners. Our contributions are for the distributed
setting, providing both algorithmic and hardness results.

Our main hardness result shows that an $\alpha$-approximation for the
minimum directed $k$-spanner problem for $k \geq 5$ requires
$\Omega(n /\sqrt{\alpha}\log{n})$ rounds using deterministic algorithms or
$\Omega(\sqrt{n }/\sqrt{\alpha}\log{n})$ rounds using randomized ones,
in the \congest model of distributed computing. 
Combined with the constant-round $O(n^{\epsilon})$-approximation algorithm 
in the \local model of [Barenboim, Elkin and Gavoille, 2016], as well as a
polylog-round $(1+\epsilon)$-approximation algorithm in the \local model that
we show here, our lower bounds for the \congest model imply a strict separation
between the \local and \congest models. Notably, to the best of our knowledge, 
this is the first separation between these models for a 
\emph{local approximation problem}.

Similarly, a separation between the directed and undirected cases is
implied. We also prove that the minimum \emph{weighted} $k$-spanner problem
for $k \geq 4$ requires a near-linear number of rounds in the \congest model, 
for directed or undirected graphs. In addition, we show lower bounds for the
minimum weighted 2-spanner problem in the \congest and \local models.

On the algorithmic side, apart from the aforementioned $(1+\epsilon)$-approximation 
algorithm for minimum $k$-spanners, our main contribution is a new distributed 
construction of minimum 2-spanners that uses only polynomial local computations. 
Our algorithm has a \emph{guaranteed} approximation ratio of $O(\log(m/n))$ for a 
graph with $n$ vertices and $m$ edges, which matches the best known ratio for 
polynomial time sequential algorithms [Kortsarz and Peleg, 1994], and is tight if
we restrict ourselves to polynomial local computations. An algorithm with this 
approximation factor was not previously known for the distributed setting. 
The number of rounds required for our algorithm is $O(\log{n}\log{\Delta})$ w.h.p, 
where $\Delta$ is the maximum degree in the graph. Our approach
allows us to extend our algorithm to work also for the \emph{directed},
\emph{weighted}, and \emph{client-server} variants of the problem. It also provides
a \congest algorithm for the minimum dominating set problem, with a \emph{guaranteed}
$O(\log{\Delta})$ approximation ratio.

\end{abstract}

\thispagestyle{empty}
\end{titlepage}

\section{Introduction}
A \emph{$k$-spanner} of a graph $G$ is a sparse subgraph of $G$ that preserves distances up to a multiplicative factor of $k$.
First introduced in the late 80's~\cite{peleg1989optimal, peleg1989graph}, spanners have been central for numerous applications, such as synchronization~\cite{peleg1989optimal, awerbuch1992adapting, awerbuch1990network}, compact routing tables~\cite{awerbuch1992routing,peleg1989trade, thorup2001compact, chechik2013compact}, distance oracles~\cite{thorup2005approximate, baswana2006approximate, roditty2005deterministic}, approximate shortest paths~\cite{elkin2005computing, elkin2004efficient}, and more.

Due to the prominence of spanners for many distributed applications, it is vital to have distributed algorithms for constructing them. Indeed, there are many efficient distributed algorithms for finding sparse spanners in undirected graphs, which give a \emph{global guarantee} on the size of the spanner. A prime example are algorithms that construct $(2k-1)$-spanners with $O(n^{1+1/k})$ edges, for a graph with $n$ vertices \cite{derbel2008locality, baswana2007simple, elkin2017efficient, derbel2010sublinear, GrossmanParter},
which is optimal in the worst case assuming Erd{\H{o}}s's \textit{girth conjecture} \cite{erdos1964extremal}.

As opposed to finding spanners with the best worst-case sparsity, this paper focuses on the network design problem of approximating the \emph{minimum} $k$-spanner, which is a fundamental optimization problem. This is particularly crucial for cases in which the worst-case sparsity is $\Theta(n^2)$ such as $2$-spanners (complete bipartite graphs) or directed spanners.
Spanner approximation is at the heart of a rich line of recent work in the sequential setting, presenting approximation algorithms~\cite{berman2013approximation, dinitz2016approximating, chlamtavc2017approximating, dinitz2011fault, chlamtac2012everywhere}, as well as hardness of approximation results~\cite{kortsarz2001hardness, dinitz2016label, elkin2007hardness}.

There are only few distributed spanner approximation algorithms known to date.
A distributed algorithm with an expected approximation ratio of $O(\log{n})$ for the minimum 2-spanner problem is given in~\cite{dinitz2011fault}. This was recently extended to $k>2$, achieving an approximation ratio of $\widetilde{O}(\sqrt{n})$ for directed $k$-spanners~\cite{dinitz2017approximating}, which matches the best approximation known in the sequential setting \cite{berman2013approximation}. Yet, in the distributed setting, it is possible to obtain better approximations if local computation is not polynomially bounded. A \emph{constant time} $O(n^{\epsilon})$-approximation algorithm for directed or undirected minimum $k$-spanner, which takes $exp(O(1/\epsilon))+O(k)$ rounds for any constant $\epsilon > 0$ and a positive integer $k$, is given in~\cite{barenboim2016fast}. In addition, we show a polylogarithmic time $(1+\epsilon)$-approximation algorithm for these problems, following the framework of a recent algorithm for covering problems \cite{ghaffari2017complexity} (see Section \ref{sec:epsilon}). This approximation is much better than the best approximation that can be acheived in the sequential setting, due to the hardness results of \cite{dinitz2016label, elkin2007hardness}. All these algorithms work in the classic \local model of distributed computing~\cite{Linial92}, where vertices exchange messages \emph{of unbounded size} in synchronous rounds.  

A natural question is whether we can obtain good approximations efficiently also in the \congest model~\cite{peleg2000distributed}, where the messages exchanged are bounded by $O(\log{n})$ bits.
In the undirected case, efficient constructions of $(2k-1)$-spanners with $O(n^{1+1/k})$ edges in the \congest model \cite{elkin2017efficient, baswana2007simple} imply $O(n^{1/k})$-approximations, since any spanner of a connected graph has at least $n-1$ edges. However, for directed graphs there are no efficient algorithms in the \congest model.

Our contribution in this paper is twofold.
We provide the first hardness of approximation results for minimum $k$-spanners in the distributed setting. Our main hardness result shows that there are \emph{no} efficient approximation algorithms for the directed $k$-spanner problem for $k \geq 5$ in the \congest model. This explains why all the current approximation algorithms for the problem require large messages, and also creates a strict separation between the directed and undirected variants of the problem, as the latter admits efficient approximations in the \congest model. In addition, we provide new distributed algorithms for approximating the minimum $k$-spanner problem and several variants in the \local model. Our main algorithmic contributaion is an algorithm for minimum 2-spanners that uses only polynomial local computations and \emph{guarantees} an approximation ratio of $O(\log{\frac{m}{n}})$, which matches the best known approximation for polynomial sequential algorithms \cite{kortsarz1994generating}.
On the way to obtaining our results, we develop new techniques, both algorithmically and for obtaining our lower bounds, which can potentially find use in studying various related problems.

\subsection{Our contributions}

\subsubsection{Hardness of approximation}

We show several negative results implying hardness of approximating various spanner problems in both the \local and \congest models.
While there are many recent hardness of approximation results for spanner problems in the sequential setting~\cite{kortsarz2001hardness, dinitz2016label, elkin2007hardness, chlamtavc2017approximating}, to the best of our knowledge ours are the first for the distributed setting.

\paragraph{(I.) Directed $k$-spanner for $k \geq 5$ in the {\textsc{CONGEST}\xspace} model:}
Perhaps our main negative result is a proof for the hardness of approximating the directed $k$-spanner problem for $k \geq 5$ in the \congest model.
\begin{restatable}{theorem}{DirectedHardness}
\label{hardness}
Any (perhaps randomized) distributed $\alpha$-approximation algorithm in the \congest model for the directed $k$-spanner problem for $k \geq 5$ takes $\Omega(\frac{\sqrt{n}}{\sqrt{\alpha} \cdot \log{n}})$ rounds, for $1 \leq \alpha \leq \frac{n}{100}$.
\end{restatable}

When restricting attention to deterministic algorithms, we prove a stronger lower bound of $\widetilde{\Omega}({\frac{n}{\sqrt{\alpha}}})$, for any $\alpha \leq \frac{n}{c}$ for a constant $c>1$.

For example, this gives that a constant or a polylogarithmic approximation ratio for the directed $k$-spanner problem in the \congest model requires $\widetilde{\Omega}(\sqrt{n})$ rounds using randomized algorithms or $\widetilde{\Omega}(n)$ rounds using deterministic algorithms. Even an approximation ratio of only $n^{\epsilon}$ is hard and requires $\widetilde{\Omega}(n^{1/2-\epsilon/2})$ rounds using randomized algorithms or $\widetilde{\Omega}(n^{1-\epsilon/2})$ rounds using deterministic ones, for any $0 < \epsilon < 1$.
Moreover, in the deterministic case, even an approximation ratio of $\frac{n}{c}$, for appropriate values of $c$, requires $\widetilde{\Omega}(\sqrt{n})$ rounds. This is to be contrasted with an approximation of $n$, which can be obtained without any communication by taking the entire graph, since any $k$-spanner has at least $n-1$ edges.

\textbf{LOCAL vs.~CONGEST.} The major implication of the above is a strict separation between the \local and \congest models, since the former admits a constant-round $O(n^{\epsilon})$-approximation algorithm~\cite{barenboim2016fast}\footnote{In \cite{barenboim2016fast}, a constant time randomized algorithm for directed $k$-spanner is presented. However, the deterministic network decomposition presented in \cite{barenboim2016fast} gives a polylogarithmic deterministic approximation for directed $k$-spanner as well, which shows the separation also for the deterministic case.} and a polylogarithmic $(1+\epsilon)$-approximation algorithm (see Section \ref{sec:epsilon}) for directed $k$-spanners. Such a separation was previously known only for \emph{global problems} (problems that are subject to an $\Omega(D)$ lower bound, where $D$ is the diameter of the graph), and for \emph{local decision problems} (such as determining whether the graph contains a $k$-cycle). To the best of our knowledge, ours is the first separation for a \emph{local approximation problem}.

\textbf{Directed vs.~undirected.} Our lower bound also separates the undirected and directed $k$-spanner problems, since there are efficient algorithms in the \congest model for constructing $(2k-1)$-spanners with $O(n^{1+1/k})$ edges \cite{elkin2017efficient, GrossmanParter} which imply an $O(n^{1/k})$-approximation. The best randomized algorithm for the task takes $k$ rounds \cite{elkin2017efficient}, and the best deterministic algorithm is a recent algorithm which takes $O(n^{1/2-1/k})$ rounds for a constant even $k$ \cite{GrossmanParter}. Achieving the same approximation for directed graphs necessitates $\widetilde{\Omega}(n^{1/2-1/{2k}})$ rounds using randomization, or $\widetilde{\Omega}(n^{1-1/{2k}})$ rounds using deterministic algorithms.

\paragraph{(II.) Weighted $k$-spanner for $k \geq 4$ in the CONGEST model:}
In addition to the above main result, we consider weighted $k$-spanners, and show that any $\alpha$-approximation for the weighted undirected $k$-spanner problem for $k \geq 4$ requires $\widetilde{\Omega}(\frac{n}{k})$ rounds, and that $\widetilde{\Omega}(n)$ rounds are needed for the weighted directed $k$-spanner problem.

\textbf{Weighted vs.~unweighted.}
As these lower bounds hold also for randomized algorithms, we obtain yet another separation, between the weighted and the unweighted variants of the problem, since the aforementioned $k$-round $(2k-1)$-spanner constructions imply an $O(n^{1/k})$ approximation for the unweighted case.

\textbf{LOCAL vs.~CONGEST.}
Since both the constant-round algorithm for approximating $k$-spanners within a factor of $O(n^{\epsilon})$~\cite{barenboim2016fast} and the $(1+\epsilon)$-approximation algorithm that we give in Section \ref{sec:epsilon} are suitable for the weighted case, our hardness result for the weighted case implies the separation between the \local and \congest models also when having weights. This holds also for the undirected weighted case.

\paragraph{(III.) Weighted $2$-spanner in the LOCAL and CONGEST models:}
Finally, we show lower bounds for the weighted $2$-spanner problem, which, in a nutshell, are obtained by a reduction that captures the intuition that approximating the minimum weight $2$-spanner is at least as hard as approximating the minimum vertex cover (MVC).
We emphasize that the reduction from the set cover problem to the unweighted 2-spanner problem given in~\cite{kortsarz2001hardness} is inherently sequential, by requiring the addition of a vertex that is connected to all other vertices in the graph, and hence is unsuitable for the distributed setting.

Our reduction implies that $\Omega(\frac{\log{\Delta}}{\log{\log{\Delta}}})$ or  $\Omega(\sqrt{\frac{\log{n}}{\log{\log{n}}}})$ rounds are required for a logarithmic approximation ratio for weighted 2-spanner in the \local model, by plugging in the lower bounds for MVC given in~\cite{kuhn2016local}. 
In addition, our reduction implies an $\widetilde{\Omega}(n^2)$ lower bound for an exact solution for weighted 2-spanner in the \congest model, by using the
\emph{near-quadratic} lower bound for exact MVC given recently in~\cite{censor2017quadratic}.
This is tight up to logarithmic factors since $O(n^2)$ rounds allow learning the entire graph topology and solving essentially all natural graph problems.

\subsubsection{Distributed approximation algorithms}

We show new distributed algorithms for approximating minimum $k$-spanners. Our main algorithmic contribution is a new algorithm for the minimum 2-spanner problem that uses only polynomial local computations (see Section \ref{sec:alg}). In addition, we show that if local computation is not polynomially bounded it is possible to achieve $(1+\epsilon)$-approximation for minimum $k$-spanners (see Section \ref{sec:epsilon}).

\paragraph{{(I.) Distributed $(1+\epsilon)$-approximation of minimum $k$-spanners:}}

In Section \ref{sec:epsilon}, we present $(1+\epsilon)$-approximation algorithms for spanner problems, following the framework of a recent algorithm for covering problems \cite{ghaffari2017complexity}. We show the following.

\begin{restatable}{theorem}{epsilonAlg} \label{epsilonAlg}
There is a randomized algorithm with complexity $O(poly(\log{n}/\epsilon))$ in the \local model that computes a $(1+\epsilon)$-approximation of the minimum $k$-spanner w.h.p, where $k$ is a constant.
\end{restatable}

The algorithm is quite general and can be adapted similarly to additional variants.
Theorem \ref{epsilonAlg} shows that although spanner problems are hard to approximate in the sequential setting, it is possible to achieve extremely strong approximations for them efficiently in the \local model. This demonstrates the power of the \local model. However, the algorithm is based on learning neighborhoods of polylogarithmic size and solving NP-complete problems (finding optimal spanners). It is desirable to design also algorithms that work with more realistic assumptions. We next focus on the 2-spanner problem and show a new algorithm that uses only polynomial local computations and uses the power of the \local model only for learning neighborhoods of diameter 2.

\paragraph{(II.) Distributed approximation of minimum $2$-spanners:}
If we restrict ourselves to polynomial local computations, the best algorithm for the minimum 2-spanner problem is the $O(\log{n})$-round $O(\log{n})$-approximation in expectation of Dinitz and Krauthgamer \cite{dinitz2011fault},\footnote{In \cite{dinitz2011fault}, a time complexity of $O(\log^2{n})$ rounds is claimed. However, the algorithm is based on sampling a certain decomposition $O(\log{n})$ times independently, which takes $O(\log{n})$ rounds each time. From the independence of the decompositions, the computations can be parallelized in the \local model, achieving a time complexity of $O(\log{n})$ rounds. See also \cite{dinitz2017approximating}.} which solves even the more general problem of finding fault-tolerant spanners.

However, this still leaves several open questions regarding minimum 2-spanners. First, the best approximation to the problem in the sequential setting is $O(\log \frac{m}{n})$ where $m$ is the number of edges in the graph. Can we achieve such approximation also in the distributed setting? Second, the approximation ratio holds only in expectation. Can we design an algorithm that \emph{guarantees} the approximation ratio? Third, this algorithm requires learning neighborhoods of logarithmic radius, and hence a direct implementation of it in the \congest model is not efficient. Can we design a more efficient algorithm in the \congest model?

We design a new algorithm for the minimum $2$-spanner problem, answering some of these questions. Our algorithm obtains an approximation ratio of $O(\log \frac{m}{n})$ \emph{always}, within $O(\log{n} \log{\Delta})$ rounds w.h.p,\footnote{As standard in this setting, a high probability refers to a probability that is at least $1-\frac{1}{n^c}$ for a constant $c \geq 1$.} where $\Delta$ is the maximum vertex degree, summarized as follows.

\begin{restatable}{theorem}{minimumSpanner} \label{minimumSpanner}
There is a distributed algorithm for the minimum 2-spanner problem in the \local model that guarantees an approximation ratio of $O(\log{\frac{m}{n}})$, and takes $O(\log{n} \log{\Delta})$ rounds w.h.p.
\end{restatable}

Our approximation ratio of $O(\log{\frac{m}{n}})$ matches that of the best approximation in the sequential setting up to a constant factor~\cite{kortsarz1994generating}, and is tight if we restrict ourselves to polynomial local computations \cite{kortsarz2001hardness}.
In addition, the approximation ratio of our algorithm is guaranteed, rather than only holding in expectation. This is crucial for the distributed setting since, as opposed to the sequential setting, running the algorithm several times and choosing the best solution completely blows up the complexity because learning the cost of the solution requires collecting global information. 
Note that although our algorithm can be converted into an algorithm with a $\emph{guaranteed}$ polylogarithmic time complexity and an approximation ratio that holds only in expectation, the opposite does not hold.
Another feature of our algorithm is that it uses the power of the \local model only for learning the 2-neighborhood of vertices. A direct implementation of our algorithm in the \congest model yields an overhead of $O(\Delta)$ rounds, which is efficient for small values of $\Delta$. We address this issue further in Section~\ref{discussion}.

\paragraph{(III.) Distributed approximation of additional $2$-spanners:}
The techniques we develop for constructing and analyzing our spanner have the advantage of allowing us to easily extend our construction to the \emph{directed}, \emph{weighted} and \emph{client-server} variants of the problem. We obtain the same approximation ratio for the directed case as in the undirected case, and for the weighted case we give an approximation ratio of $O(\log \Delta)$, both improving upon the $O(\log{n})$ approximation in expectation of \cite{dinitz2011fault}. For the client-server 2-spanner case, which to the best of our knowledge ours is the first distributed approximation, we obtain an approximation ratio that matches that of the sequential algorithm~\cite{elkin1999client}.

\paragraph{(IV.) Distributed approximation of MDS:}
Finally, our technique also gives an efficient algorithm for the minimum dominating set (MDS) problem, which obtains an approximation ratio of $O(\log{\Delta})$ \emph{always}. Our algorithm for MDS works even in the \congest model and takes $O(\log{n}\log{\Delta})$ rounds w.h.p. The MDS problem has been studied extensively by the distributed computing community, with several efficient algorithms for MDS in the \congest obtaining an approximation ratio of $O(\log{\Delta})$ in expectation \cite{kuhn2003constant,kuhn2016local,jia2002efficient}. To the best of our knowledge, our algorithm is the first that guarantees this approximation ratio always.

\subsection{Technical overview}

\subsubsection{Hardness of approximation} \label{sec:intro}

We prove Theorem~\ref{hardness} by a reduction from 2-party communication problems, as has been proven fruitful for various lower bounds for the \congest model~\cite{censor2017quadratic, Abboud2016, sarma2012distributed, drucker2014power, frischknecht2012networks, holzer2012optimal}. In principle, a family of graphs is constructed depending on the input strings of the two players, such that the solution to the required \congest problem uniquely determines whether the input strings of the players satisfy a certain Boolean predicate. The most common usage is of set-disjointness, although other 2-party communication problems have been used as well~\cite{PelegR00, Elkin06, fischer2017distributed, censor2016distributed}. The two players can simulate a distributed algorithm for solving the \congest problem, and deduce their output for the 2-party communication problem accordingly. This yields a lower bound for the \congest problem, based on known lower bounds for the communication complexity of the 2-party problem, by incorporating the cost of the simulation itself.

The prime caveat in using this framework for approximation problems is that in the above examples a modification of a single input bit has a slight influence on the graph. For example, when showing a lower bound for computing the diameter, any bit of the input affects the distance between one pair of vertices \cite{Abboud2016, frischknecht2012networks, holzer2012optimal}.
This is sufficient when computing some global property of the graph. Indeed, the distance between a single pair of vertices can change the diameter of the graph.
The challenge in designing a construction for approximating $k$-spanners is that now any single bit needs to affect drastically the size of the minimum $k$-spanner. In more detail, any $k$-spanner has at least $n-1$ edges and, hence, for a meaningful lower bound for an $\alpha$-approximation, any input bit must affect at least $\Omega(\alpha n)$ edges.

We manage to overcome the above challenge by constructing a graph that captures this requirement and allows a reduction from set-disjointness. The main technical ingredient is a dense component in which many edges are affected by single input bits. This component resides in its entirety within the set of vertices that is simulated by a single player of the two, thus resulting in a non-symmetric graph construction. This is crucial for our proof, as otherwise the density of this component would imply a dense cut between the two sets of vertices simulated by the players, which in turn would nullify the achievable lower bound. For having this property, we believe that our construction may give rise to follow-up lower bound constructions for additional local approximation problems.

Our graph construction is designed using several parameters, which allows us to show trade-offs between the time complexity of an algorithm and its approximation ratio, and gives lower bounds even for large values of $\alpha$.

Our stronger lower bounds for the deterministic case are obtained using the 2-party gap-disjointness problem rather than the more common set-disjointness problem. Since gap-disjointness allows more slack, we obtain stronger lower bounds, at the price of them holding only for deterministic algorithms. We believe that the flexibility of the gap-disjointness problem may be useful in showing additional strong lower bounds for approximation problems. Our stronger lower bounds for the weighted case are obtained by assigning weights to the edges of the graph in a manner which allows us to shave off certain edges that affect the bound.

\subsubsection{Distributed approximation of minimum $2$-spanners}

Our algorithm for approximating minimum $2$-spanners is inspired by the sequential greedy algorithm of Kortsarz and Peleg~\cite{kortsarz1994generating}, in which \emph{dense stars} are added to the spanner one by one, obtaining an approximation ratio of $O(\log{\frac{m}{n}})$. A \textit{star} is a subset of edges between a vertex $v$ and \emph{some} of its neighbors. The \textit{density} of a star is the ratio between the number of edges \textit{2-spanned} by the star and the size of the star, where an edge $e=\{u,v\}$ is \textit{2-spanned} by a star $S$ if $S$ includes a path of length two between $u$ and $v$. A roughly intuition for the greedy algorithm is that if $S$ is a dense star then adding its edges to the spanner allows 2-spanning many edges by adding only a small number of edges to the spanner.

A direct implementation of this greedy approach in the distributed setting is highly expensive, since deciding upon the densest star inherently requires collecting global information. Moreover, one would like to leverage the ability of the distributed setting to add multiple stars to the spanner simultaneously.
To address both sources of inefficiency, rather than computing the star that is the densest in the entire graph, we compute all the stars that are the densest in their local 2-neighborhood.
While greatly speeding up the running time, adding all of these locally densest stars to the spanner is too extreme, and results in a poor approximation ratio. Instead, we consider these stars as \textit{candidates} for being added to the spanner.

The key challenge is then to break symmetry among the candidates, while balancing the need to choose many stars in parallel (for a fast running time) with the need to bound the overlap in spanned edges among the candidates (for a small approximation ratio). We tackle this conflict by constructing a voting scheme for breaking symmetry by choosing among the stars based on a random permutation. Interestingly, our approach is inspired by a parallel algorithm for set cover~\cite{rajagopalan1998primal}. We let each edge vote for the first candidate that 2-spans it according to the random permutation. A candidate that receives a number of votes which is at least $\frac{1}{8}$ of the edges it 2-spans is added to the spanner, and we continue this process iteratively.

Since we add to the spanner only stars receiving many votes, this approach guarantees that there is not too much overlap in the edges 2-spanned by different stars, which eventually culminates in a proof of an approximation ratio of $O(\log{\frac{m}{n}})$, which matches the one obtained by the greedy approach.

A tricky obstacle lies in showing that our algorithm completes in $O(\log{n}\log{\Delta})$ rounds w.h.p. This is because, as opposed to the set cover case, there may be as many as $2^{\Delta}$ different stars centered at each vertex, and a vertex may be required to add candidate stars multiple times during the execution of the algorithm.
It turns out that an arbitrary choice for a candidate among all densest stars centered at a vertex is incapable of providing an efficient time complexity. To overcome this issue, we design a subtle mechanism for proposing a candidate star, and pair it with a proof that our algorithm indeed completes in the claimed number of rounds.

\subsection{Discussion} \label{discussion}

While our results in this paper significantly advance the state-of-the-art
in distributed approximation of minimum $k$-spanners,
intriguing questions remain open. First, the landscape of the trade-offs
between the approximation ratio and the running time of distributed minimum
$k$-spanner algorithms is yet to be fully mapped. For example, the $O(\log{\Delta})$
factor in the running time of our approximation algorithm for weighted 2-spanner is tight
up to an $O(\log{\log{\Delta}})$ factor, due to our reduction from MVC and the known
lower bounds for it. However, it remains open whether the $O(\log{n})$ factor is
necessary. Additional gaps remain open for other various approximation
ratios. In particular, an interesting question is to show a lower bound for approximating
the undirected unweighted minimum $k$-spanner problem.

A curious question is whether our algorithm can be efficiently made to work
in the \congest model. A direct implementation would yield an overhead of
$O(\Delta)$ for the running time, for computing the densities of stars, and for sending the candidate stars.
We emphasize that knowing the density of the neighborhood of vertices
is crucial for additional algorithms, such as the state-of-the-art $(\Delta+1)$-coloring
algorithm of Harris et al. \cite{harris2016distributed}. Another interesting question is to design an efficient \emph{deterministic} algorithm achieving the same approximation ratio.

For larger values of the stretch $k$, our lower bounds imply a strict
separation between the \local and \congest models for the number of rounds
required for approximating directed minimum $k$-spanners. Such a separation was
previously known only for \emph{global problems} (problems that
are subject to an $\Omega(D)$ lower bound, where $D$ is the diameter of the
graph), and for \emph{local decision problems} (such as determining
whether the graph contains a $k$-cycle). Interestingly, ours is the first
separation for a \emph{local approximation problem}. It is a central open
question whether such separations hold also for \emph{local symmetry
breaking problems}.

Interestingly, our algorithm, as well as other distributed approximation algorithms for the minimum $k$-spanner in the \local model, work also for directed graphs, achieving the same approximation ratio and round complexity. However, our hardness results create a strict separation between the undirected and directed variants in the \congest model. It will be interesting to show such separations for other problems.

\subsection{Additional related work}

Spanners have been studied extensively in the distributed setting, producing many efficient algorithms for finding sparse spanners in undirected graphs~\cite{derbel2008locality, baswana2007simple, elkin2017efficient, derbel2010sublinear, elkin2007near, GrossmanParter}.
These algorithms construct $(2k-1)$-spanners with $O(n^{1+1/k})$ edges for any fixed $k \geq 2$, with the fastest completing in $k$ rounds~\cite{derbel2008locality,elkin2017efficient}, which is tight~\cite{derbel2008locality}. Many additional works construct various non-multiplicative spanners in the distributed setting, such as~\cite{pettie2010distributed} and the excellent overview within.

Many recent studies address spanner approximations in the sequential setting. The greedy algorithm of~\cite{kortsarz1994generating} achieves an approximation ratio of $O(\log{\frac{m}{n}})$ for the minimum $2$-spanner problem. This was extended to the weighted, directed and client-server cases \cite{elkin1999client, kortsarz2001hardness}. Approximation algorithms for the directed $k$-spanner problem for $k>2$ are given in ~\cite{elkin2005approximating, berman2010finding, berman2013approximation, dinitz2011directed, dinitz2016approximating}, with the best approximation ratio of $O(\sqrt{n} \log{n})$ for $k>4$, and an approximation ratio of $\widetilde{O}(n^{1/3})$ for $k=3,4$ \cite{berman2013approximation, dinitz2016approximating}. These approximation ratios are matched by a recent distributed $O(k\log{n})$-round algorithm, that uses only polynomial local computations \cite{dinitz2017approximating}.
Approximation algorithms are given also for pairwise spanners and distance
preservers~\cite{chlamtavc2017approximating}, for spanners with lowest maximum degree~\cite{kortsarz1998generating, chlamtac2012everywhere, chlamtac2014lowest, dinitz2017approximating}, for fault-tolerant spanners~\cite{dinitz2011fault, dinitz2016approximating}, and more.

Hardness of approximation results in the sequential setting give that for $k=2$, no polynomial algorithm gives an approximation ratio better than $\Theta(\log{n})$~\cite{kortsarz2001hardness}, which shows that the sequential greedy algorithm is optimal. For $k>2$, the problem is even harder. For any constant $\epsilon > 0$ and $k \geq 3$ there are no polynomial-time algorithms that approximate the $k$-spanner problem within a factor better than $2^{(\log^{1-\epsilon}{n})/k}$~\cite{dinitz2016label}, or the directed $k$-spanner problem within a factor better than $2^{(\log^{1-\epsilon}{n})}$~\cite{elkin2007hardness}. Similar results are known for additional variants ~\cite{elkin2007hardness,chlamtavc2017approximating}.

Spanner problems are closely related to covering problems such as set cover, minimum dominating set (MDS), and minimum vertex cover. Indeed, some of the ingredients of our algorithms borrow ideas from distributed and parallel algorithms for such problems. Our symmetry breaking scheme is inspired by the parallel algorithm for set cover of Rajagopalan and Vazirani~\cite{rajagopalan1998primal}, however, the general structure of this algorithm requires global coordination and hence is not suitable for the distributed setting.
There are also several ideas inspired by the distributed MDS algorithm of Jia et al.~\cite{jia2002efficient}, such as, rounding the densities and comparing densesties in 2-neighborhoods. However, \cite{jia2002efficient} breaks the symmetry between the candidates in a different way which results in an approximation ratio of $O(\log{\Delta})$ in expectation. The connection between spanners to set cover is used also in \cite{berman2010finding} where they show that covering the edges of a graph by stars is also useful for approximating the directed $k$-spanner problem for $k>2$.
In this context, we also mention the distributed algorithm of~\cite{Ghaffari14} for the minimum connected dominating set problem, which also uses stars as the main component for its construction. Our work is, however, incomparable, especially since the minimum connected dominating set problem is a \emph{global} problem, admitting an $\Omega(D)$ lower bound even in the \local model.

\subsection{Preliminaries}

Let $G=(V,E)$ be a connected undirected graph with $n$ vertices and maximum degree $\Delta$. Let $S \subseteq E$ be a subset of the edges, and let $k \geq 1$. We say that an edge $e=\{u,v\}$ is \textit{covered} by $S$ if there is a path of length at most $k$ between $u$ and $v$ in $S$.
A \textit{k-spanner} of $G$ is a subgraph of $G$ that covers all the edges of $G$. A $k$-spanner of a subgraph $G' \subseteq G$ is a subgraph of $G$ that covers all the edges of $G'$.
For a directed graph, we say that a directed edge $e=(u,v)$ is covered by a subset of edges $S$, if $S$ includes a \textit{directed} path of length at most $k$ from $u$ to $v$, and define a $k$-spanner for a directed graph accordingly.

In the \emph{minimum $k$-spanner} problem the input is a connected undirected graph $G=(V,E)$ and the goal is to find the minimum size $k$-spanner of $G$.
The \emph{directed $k$-spanner} problem is defined accordingly, with respect to directed graphs. In the \emph{weighted $k$-spanner} problem each edge $e$ has a non-negative weight $w(e)$ and the goal is to find the $k$-spanner of $G$ having minimum cost, where the cost of a spanner $H$ is $w(H) = \sum_{e \in H} w(e).$\footnote{There is another variant of the weighted $k$-spanner problem, in which the weight of an edge represents a length. We emphasize that in our case all the edges have length 1.}
In the \emph{client-server $k$-spanner} problem, introduced in \cite{elkin1999client}, the input is a connected undirected graph $G=(V,E)$ that its edges are divided to two types: clients $C$ and servers $S$ (there may be edges $e \in C \cap S)$, and the goal it to find the minimum size $k$-spanner of $C$ that includes only edges of $S$.

In the distributed setting, the input for the $k$-spanner problem is the communication graph $G$ itself. Each vertex initially knows only the identities of its neighbors, and needs to output a subset of its edges such that the union of all outputs is a $k$-spanner. The communication in the network is \emph{bidirectional}, even when solving the directed $k$-spanner problem.

\paragraph*{Roadmap:} In Section \ref{sec:hardness}, we present our hardness of approximation results for directed and weighted $k$-spanners in the \congest model. In Section \ref{sec:MVC}, we provide hardness of approximation results for weighted 2-spanners. In Section \ref{sec:alg}, we present our algorithm for the minimum 2-spanner problem and show its extensions to other variants. In Section \ref{sec:MDS}, we describe our MDS algorithm. Finally, in Section \ref{sec:epsilon}, we show our $(1+\epsilon)$-approximation for minimum $k$-spanners.

\section{Hardness of approximation in the CONGEST model} 
\label{sec:hardness}

In this section, we prove hardness of approximation results for approximating $k$-spanners in the \congest model. As explained in Section~\ref{sec:intro}, we build upon the previous used framework of reducing 2-party communication problems to distributed problems for the \congest model. The key technical challenge that we overcome is how to plant a dense subgraph into the construction, without inducing a large cut between the vertices simulated by the two players, but while still having the choice of edges taken from the dense subgraph to the spanner depend on both inputs.

We describe a graph construction that allows us to provide a reduction from problems of 2-party communication. In the latter setting, two players, Alice and Bob, receive input strings $a=(a_1,...,a_N)$ and $b=(b_1,...,b_N)$, respectively, of size $N$. Their goal is to solve a problem related to their inputs, while communicating a minimum number of bits. For example, the set disjointness problem requires the players to decide if their input strings represent disjoint subsets of $[N]$, that is, they need to decide if there is a bit $1 \leq i \leq N$ such that $a_i = b_i =1$. The communication complexity of set disjointness is known to be linear in the length of the strings~\cite{razborov1992distributional,Kushilevitz:1996:CC:264772}.

\begin{lemma} \label{set_disj}
Solving the set disjointness problem on input strings of size $N$, requires exchanging $\Omega(N)$ bits, even using randomized protocols.
\end{lemma}

We start by showing that approximating the \emph{directed} $k$-spanner problem in the \congest model is hard for $k \geq 5$, and then modify our construction to provide hardness results for the weighted case.

The general approach is to build a dense graph $G$, where some of its edges depend on the inputs of Alice and Bob, such that if the inputs of Alice and Bob are disjoint then there is a sparse $5$-spanner in $G$ (which is also a $k$-spanner for $k \geq 5$), and otherwise any $k$-spanner has many edges. By simulating the distributed approximation algorithm for the $k$-spanner problem, Alice and Bob solve set disjointness. Hence, depending on the parameters of our graph construction, a communication lower bound for the latter would imply a lower bound on the number of rounds required for the former.

In \cite{frischknecht2012networks,holzer2012optimal}, a reduction from set disjointness is used in order to show a lower bound for computing the diameter of a graph. The main idea is that each bit of the inputs affects the distance between two vertices in the graph, and if the distance between any of these pairs of vertices is long it affects the diameter of the graph. This idea is useful also for showing lower bound for spanner problems, and indeed one of the elements in our construction is similar to the constructions in \cite{frischknecht2012networks,holzer2012optimal}. However, the main difference in our case is that the distance between one pair of vertices in the graph does not affect significantly the size of the minimum spanner.

In order to overcome it, we suggest the following construction. Our graph consists of two subgraphs. One of them depends on the inputs, and the other one is a complete bipartite graph $D$ that each of its sides is divided to blocks of size $\beta$. We connect the two subgraphs in such a way that each bit $i$ of the inputs affects $\beta^2$ edges of $D$, which must be added to the spanner if and only if $a_i = b_i =1$.

Let $\ell, \beta$ be positive integers. We construct a graph $G=G(\ell,\beta)$ according to the parameters $\ell$ and $\beta$. Later we plug-in different values of $\ell$ and $\beta$ in order to obtain several trade-offs.
The graph $G(\ell,\beta)$ is a directed graph, with $V = X_1 \cup X_2 \cup Y_1 \cup Y_2 \cup Y_3$, where $X_1 = \{x^1_i| 1 \leq i \leq \ell \} \cup \{x^2_i| 1 \leq i \leq \ell \}, Y_1 = \{y^1_i| 1 \leq i \leq \ell\} \cup \{y^2_i| 1 \leq i \leq \ell \}$, $X_2 = \{x_{ij}|1 \leq i \leq \ell, 1 \leq j \leq \beta \}, Y_2 = \{y_{ij}|1 \leq i \leq \ell, 1 \leq j \leq \beta \}$, and $Y_3 = \{y^3_i| 1 \leq i \leq \ell\}$. See Figure \ref{alice_bob} for an illustration.

The set of edges consists of a matching between $X_1$ and $Y_1$ that includes all the directed edges $(x^1_i,y^1_i)$ and $(x^2_i,y^2_i)$, for $1 \leq i \leq \ell$. In addition, there is a complete bipartite graph $D$ between the vertices of $X_2$ and $Y_2$ that includes all the directed edges $(x_{ij},y_{rs})$ for $1 \leq i,r \leq \ell, 1 \leq j,s \leq \beta$. For each vertex $x_{ij} \in X_2$ there is an edge $(x_{ij},x^1_i)$. For each vertex $y_{ij} \in Y_2$ there is an edge $(y^3_i,y_{ij})$. In addition, the graph includes the edges $(y^2_i,y^3_i)$, for $1 \leq i \leq \ell$.

In addition, the two input strings $a,b$ of length $\ell^2$ bits, denoted by $a_{ij},b_{ij}$ for $1 \leq i,j \leq \ell$, affect $G$ in the following way. The edge $(x^1_i,x^2_j)$ is in $G$ if and only if $a_{ij}=0$, and the edge $(y^1_i,y^2_j)$ is in $G$ if and only if $b_{ij}=0$.

\setlength{\intextsep}{0pt}
\begin{figure}[h]
\centering
\setlength{\abovecaptionskip}{-2pt}
\setlength{\belowcaptionskip}{8pt}
\includegraphics[scale=0.5]{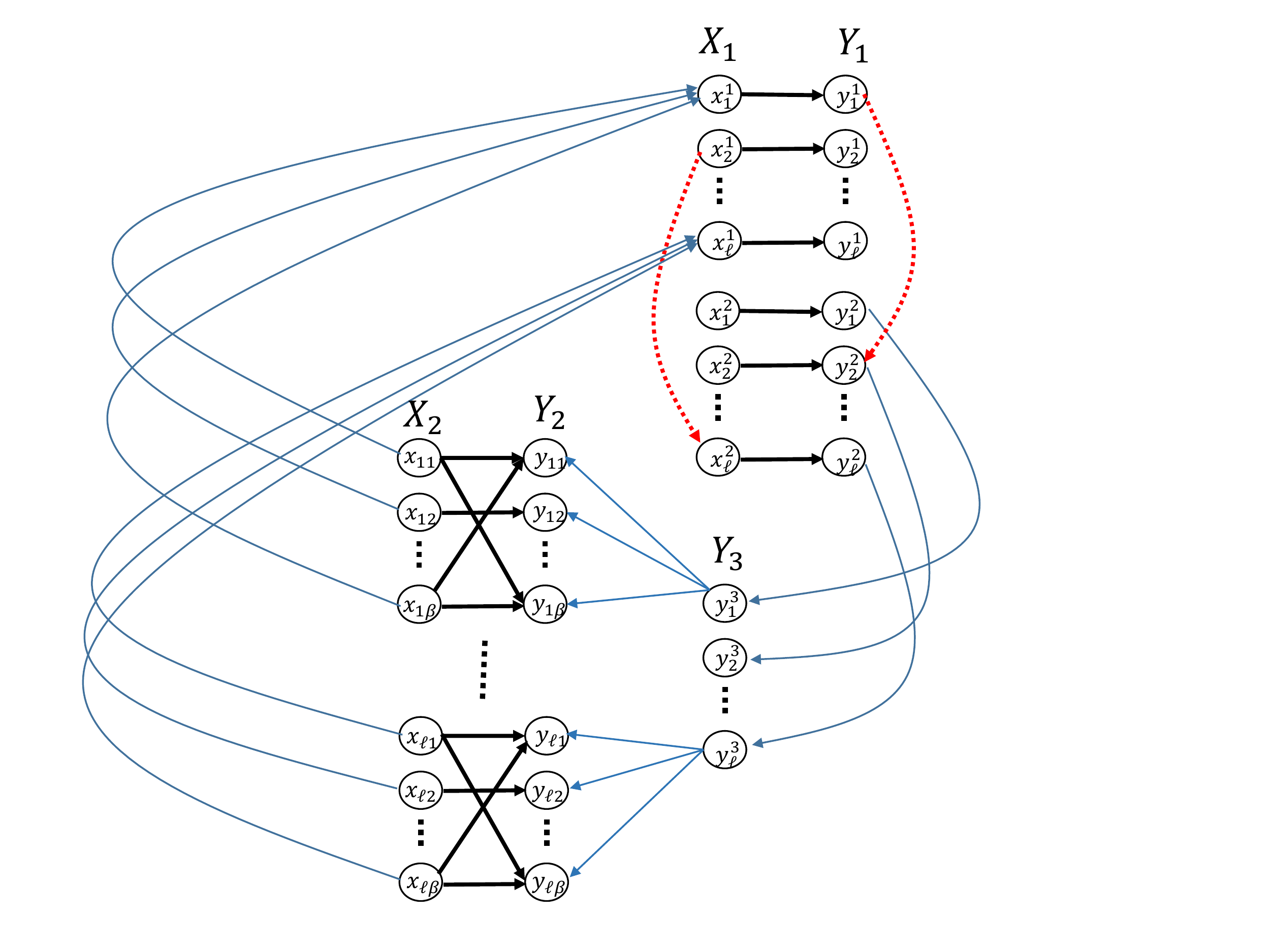}
 \caption{The graph $G$, with some of its edges omitted for clarity. The red dashed edges are examples of optional edges which depend on the input strings $a$ and $b$.}
\label{alice_bob}
\end{figure}

Note that the number of vertices in $G$ is $n = \Theta(\ell \beta)$, and that $D$ consists of $(\ell \beta)^2 = \Theta(n^2)$ edges, and recall the goal of constructing a sparse $k$-spanner for $G$ with $k \geq 5$.
Since $D$ is a dense subgraph, taking its edges to the spanner would be expensive. However, in order to avoid taking the edges of $D$ to the spanner, the spanner must include a directed path of length at most $k$ between every pair of vertices $x_{ij},y_{rs}$, which does not include edges of $D$. The existence of such a path depends on the input strings in the following way.

\begin{claim} \label{directed}
If one of the edges $(x^1_i,x^2_r),(y^1_i,y^2_r)$ is in $G$, there is a directed path of length $5$ between the vertices $x_{ij},y_{rs}$ that does not contain edges of $D$. Otherwise, the only directed path from $x_{ij}$ to $y_{rs}$ is the path that consists of the edge $(x_{ij},y_{rs})$.
\end{claim}

\begin{proof}
Note that any directed path from $x_{ij}$ to $y_{rs}$ that does not include the edges of $D$ must begin with the edge $(x_{ij},x^1_i)$ and must end with the two edges $(y^2_r,y^3_r),(y^3_{r},y_{rs})$. Hence, the existence of such a path depends on whether there is a directed path from $x^1_i$ to $y^2_r$. We show that there is a directed path of length 2 from $x^1_i$ to $y^2_r$ if at least one of the edges $(x^1_i,x^2_r),(y^1_i,y^2_r)$ is in $G$. Otherwise, there is no directed path of any length from $x^1_i$ to $y^2_r$.

Let $P$ be a directed path from $x^1_i$ to $y^2_r$. The path $P$ must cross the cut between $X_1$ to $Y_1$ either by the edge $(x^1_i,y^1_i)$ or by the edge $(x^2_r,y^2_r)$, since any path (of any length) from $x^1_i$ can only cross the cut through the edge $(x^1_i,y^1_i)$ or by an edge of the form $(x^2_j,y^2_j)$. However, if $j \neq r$, $y^2_r$ is not reachable from $y^2_j$. If $P$ crosses by the edge $(x^1_i,y^1_i)$ the only way to reach $y^2_r$ from $y^1_i$ is by the edge $(y^1_i,y^2_r)$. In the second case, the edge $(x^1_i,x^2_r)$ must be the first edge in the path.

In conclusion, if one of the edges $(x^1_i,x^2_r),(y^1_i,y^2_r)$ is in $G$, then there is a directed path of length $5$ between the vertices $x_{ij},y_{rs}$ that does not contain edges of $D$. Otherwise, there is no directed path of any length from $x^1_i$ to $y^2_r$. In this case, the only directed path from $x_{ij}$ to $y_{rs}$ is the path that consists of the edge $(x_{ij},y_{rs})$.
\end{proof}

Claim~\ref{directed} captures the essence of why our construction is suitable for an approximation problem. Next, we use our graph construction and this claim in order to show our hardness results.

\subsection{Randomized directed $k$-spanner}

In this section, we address the directed $k$-spanner problem for $k \geq 5$, and show that obtaining an $\alpha$-approximation requires $\Omega(\frac{\sqrt{n }}{\sqrt{\alpha}\log{n}})$ rounds in the \congest model, even when using randomized algorithms.

\begin{lemma} \label{disj}
Let $G = G(\ell,\beta)$ for $\beta \geq \ell$, let $k \geq 5$, and let $c=7$. If the input strings $a,b$ are disjoint, then there is a $k$-spanner of size at most $c \ell \beta$ for $G$. Otherwise, any $k$-spanner for $G$ includes at least $\beta^2$ edges of $D$.
\end{lemma}

\begin{proof}
If the input strings $a,b$ are disjoint, then for every pair of indexes $i,r$ at least one of the edges $(x^1_i,x^2_r),(y^1_i,y^2_r)$ is in $G$. Hence, by Claim \ref{directed}, there is a directed path of length at most $5$ between every two vertices $x_{ij},y_{rs}$, which does not contain edges of $D$. This gives a 5-spanner of size at most $c \ell \beta$ edges for $G$ by taking all the edges not in $D$, since there are at most $2\ell \beta + 2\ell^2+3\ell$ such edges, which is at most $c\ell\beta$ since $\ell \leq \beta$ and $c=7$. This is also a $k$-spanner for any $k \geq 5$.

If the input strings are not disjoint, then there is a pair of indexes $i,r$ such that neither of the edges $(x^1_i,x^2_r),(y^1_i,y^2_r)$ is in $G$. Hence, by Claim \ref{directed}, there is no directed path between the vertices $x_{ij},y_{rs}$ except for the path that includes the edge $(x_{ij},y_{rs})$. Therefore, we need to take all the edges $(x_{ij},y_{rs})$ to the spanner for all values of $j$ and $s$, which means adding $\beta^2$ edges of $D$ to the spanner.
\end{proof}

Let $k \geq 5$ and let $A$ be a distributed $\alpha$-approximation algorithm for the minimum $k$-spanner problem. Denote by $T(n)$ the time complexity of $A$ on a graph with $n$ vertices. The approximation ratio $\alpha=\alpha(n)$ of the algorithm $A$ may depend on $n$, and we assume that it is a monotonic increasing function of $n$, and that if $n = \Theta(n')$, then $\alpha(n) = \Theta(\alpha(n'))$.

Our goal is to show that $A$ can be used to solve set disjointness.
If $\alpha \cdot c \ell \beta < \beta^2$, then by Lemma \ref{disj}, the algorithm $A$ gives a protocol for set disjointness, in which case we show a lower bound of $\Omega(\frac{\ell}{\log{n}})$ on the time complexity of $A$, as stated in the following lemma.

\begin{lemma} \label{threshold}
Let $G=G(\ell,\beta)$. If there is a threshold $t$ such that if the input strings $a,b$ are disjoint, an optimal spanner of $G$ has at most $t$ edges, and otherwise each spanner of $G$ includes more than $\alpha(n) \cdot t$ edges of $D$, then $T(n) = \Omega(\frac{\ell}{\log{n}}).$
\end{lemma}

\begin{proof}
We use $A$ to solve set disjointness on input strings of length $N = \ell^2$ in the following way. Let $a,b$ be two input strings of length $N$, given to Alice and Bob respectively.
We take the graph $G=G(\ell,\beta)$ and define $V_B = Y_1, V_A = V \setminus V_A$. Since the input strings $a$ and $b$ affect only edges between vertices within $V_A$ and within $V_B$ respectively, it holds that Alice knows all the edges adjacent to vertices in $V_A$ and Bob knows all the edges adjacent to vertices in $V_B$.
The cut between $V_A$ to $V_B$ consists of $\Theta(\ell)$ edges: the $2\ell$ edges of the matching between $X_1$ to $Y_1$, and the $\ell$ edges between $Y_1$ to $Y_3$. Now Alice and Bob simulate $A$ on $G$ as follows. Alice simulates the vertices in $V_A$ and Bob simulates the vertices in $V_B$. At each round, Alice and Bob exchange the messages going over the cut between $V_A$ and $V_B$ in either direction. Messages that are sent between vertices in $V_A$ or between vertices in $V_B$ are simulated locally by Alice and Bob, without any communication.
Since the size of messages is $O(\log{n})$ bits, and the size of the cut is $\Theta(\ell)$, they can simulate one round of $A$ by exchanging at most $O(\ell \cdot \log{n})$ bits, and therefore they can simulate the entire execution of $A$ by exchanging at most $O(T(n) \cdot \ell \cdot \log{n})$ bits.

At the end of the simulation, Alice knows which of the edges of $D$ are taken to the spanner. If there are more than $\alpha(n) \cdot t$ edges of $D$ in the spanner, Alice concludes that the input strings are not disjoint, and otherwise she concludes that they are disjoint.

To show that this produces the correct output, recall the condition of the lemma that if the input strings are disjoint then the size of an optimal spanner is at most $t$ and otherwise it is more than $\alpha(n) \cdot t$. Therefore, if the input strings are disjoint, since $A$ is an $\alpha(n)$-approximation algorithm, it constructs a spanner with at most $\alpha(n) \cdot t$ edges, in which case Alice indeed outputs that the input strings are disjoint. Otherwise, if the input strings are not disjoint, the size of any spanner is more than $\alpha(n) \cdot t$ edges, in which case Alice indeed outputs that the input strings are not disjoint.

Hence, Alice and Bob solve set disjointness by exchanging $O(T(n) \cdot \ell \cdot \log{n})$ bits. However, any (perhaps randomized) protocol that solves disjointness on inputs of size $N= \ell^2$ requires exchanging $\Omega(\ell^2)$ bits by Lemma \ref{set_disj}. This gives $T(n) = \Omega(\frac{\ell^2}{\ell \cdot \log{n}}) = \Omega(\frac{\ell}{\log{n}})$.
\end{proof}

Using Lemma \ref{disj} and Lemma \ref{threshold}, we prove our following main theorem.

\DirectedHardness*

\begin{proof}
We show that there is a threshold $t$ that distinguishes whether the inputs are disjoint. Then, using Lemma \ref{threshold}, we get a lower bound on the round complexity of $A$.

We define $G=G(\ell,\beta)$ with the following choice of the parameters $\beta,\ell$.
Let $n'$ be a positive integer, and let $c=7$. Let $q = \lceil \alpha(n') c \rceil + 1$. Let $\ell = \lfloor \sqrt{\frac{n'}{cq}} \rfloor$, and let $\beta = q \ell$. The requirement $\alpha(n) \leq \frac{n}{100}$ ensures that $cq \leq n'$, which shows that $\ell$ is positive. The number of vertices in $G$ is $n = \Theta( \ell \beta) = \Theta(q \ell^2) = \Theta (q \cdot \frac{n'}{q})=\Theta(n').$ In addition, note that $n \leq c \ell \beta$, since the number of vertices in $G$ is $2 \ell \beta + 5 \ell$, which gives
$n \leq c \ell \beta = cq \ell^2 \leq cq \cdot \frac{n'}{cq}=n'.$

Let $t= c \ell \beta$. By Lemma \ref{disj}, If the inputs are disjoint, there is a $k$-spanner for $G$ having at most $t = c  \ell \beta$ edges, and otherwise any $k$-spanner for $G$ includes at least $\beta^2$ edges of $D$. By the definition of $q$, it holds that $\alpha(n') \cdot c < q$, which gives $\alpha(n') \cdot c \ell \beta < q \ell \beta = \beta^2 $. Since $n \leq n'$, it holds that $\alpha(n) \leq \alpha(n')$, which gives $\alpha(n) \cdot t = \alpha(n) \cdot c \ell \beta < \beta^2$.

Hence, $t$ satisfies the conditions of Lemma \ref{threshold}, which gives $T(n) = \Omega(\frac{\ell}{\log{n}})$.
Since $\ell = \Theta(\sqrt{\frac{n}{q}})=\Theta(\sqrt{\frac{n}{\alpha(n)}})$,
it holds that $T(n) = \Omega(\frac{\ell}{\log{n}}) = \Omega(\frac{\sqrt{n}}{\sqrt{\alpha(n)} \cdot \log{n}}).$
\end{proof}

Theorem \ref{hardness} shows that achieving a constant or a polylogarithmic approximation ratio for the directed $k$-spanner problem in the \congest model requires $\widetilde{\Omega}(\sqrt{n})$ rounds, and even achieving an approximation ratio of $n^{\epsilon}$ is hard, requiring $\widetilde{\Omega}(n^{1/2-\epsilon/2})$ rounds, for any $0 < \epsilon < 1$.

This proves a strict separation between the \local and \congest models, since there is a constant round  $O(n^{\epsilon})$-approximation algorithm \cite{barenboim2016fast}, and a polylogarithmic $(1+\epsilon)$-approximation algorithm (see Section \ref{sec:epsilon}) for directed $k$-spanner in the \local model.

It also separates the undirected and directed $k$-spanner problems, since there are randomized $k$-round algorithms in the \congest model for constructing $(2k-1)$-spanners with $O(n^{1+1/k})$ edges \cite{elkin2017efficient}. These algorithms obtain an approximation ratio of $O(n^{1/k})$ for the undirected minimum $(2k-1)$-spanner problem in $k$ rounds, where achieving the same approximation for the directed problem requires $\widetilde{\Omega}(n^{1/2-1/{2k}})$ rounds according to Theorem \ref{hardness}.

\subsection{Deterministic directed $k$-spanner} \label{sec:deterministic}

We next show that any \emph{deterministic} algorithm solving the directed $k$-spanner problem for $k \geq 5$, requires $\Omega(\frac{n}{\sqrt{\alpha} \cdot \log{n}})$ rounds. The trick that allows a stronger lower bound is that we use a different problem from communication complexity, which we refer to as the \textit{gap disjointness} problem. This problem is also mentioned in \cite{fischer2017distributed}.

In the gap disjointness problem, Alice and Bob receive the input strings $a=(a_1,...,a_N)$ and $b=(b_1,...,b_N)$, respectively, and their goal is to distinguish whether their input strings are disjoint or are \textit{far} from being disjoint. The inputs are far from being disjoint if there are at least $\frac{N}{12}$ indexes $i$, such that $a_i = b_i =1$. If the inputs are neither disjoint nor far from being disjoint, any output of Alice and Bob is valid.
The gap disjointness problem can be easily solved by randomized protocols exchanging $O(1)$ bits. However, solving the problem deterministically requires exchanging $\Omega(N)$ bits.

\begin{lemma} \label{gap}
Solving the gap disjointness problem deterministically on input strings of size $N$ requires exchanging $\Omega(N)$ bits.
\end{lemma}

For a proof of Lemma \ref{gap}, see example 5.5 in \cite{Kushilevitz:1996:CC:264772}, where it is shown that approximating the size of the intersection $|a \cap b|$ requires exchanging $\Omega(N)$ bits. The proof relies only on showing that distinguishing between disjoint inputs and inputs with intersection of more than $\frac{N}{6}$ bits is difficult (note that any such inputs have intersection of size at least $\frac{N}{12}$). Hence, the exact same proof shows that solving gap disjointness requires exchanging $\Omega(N)$ bits using a deterministic protocol.

In order to use set disjointness for the proof of Theorem~\ref{hardness}, it was necessary to devise a construction where each bit of the input affects many edges of the spanner, in order to argue that even if there is only one index $i$ such that $a_i = b_i =1$, then the players can correctly decide whether the inputs are disjoint by checking the size of the spanner. However, when we use gap disjointness, the players need to distinguish only between the case that the inputs are disjoint and the case that they are far from being disjoint, which allows much more flexibility and gives stronger lower bounds for the deterministic case.

\begin{lemma} \label{disj2}
Let $G=G(\ell,\beta)$ for $1 \leq \beta \leq \ell$, let $k \geq 5$ and let $c=7$.
If the input strings $a,b$ are disjoint, then there is a $k$-spanner of size at most $c \ell^2$. If the input strings are far from being disjoint, any $k$-spanner for $G$ includes at least $\frac{\beta^2}{12} \ell^2$ edges of $D$.
\end{lemma}

\begin{proof}
If the input strings are disjoint, taking all the edges not in $D$ is a 5-spanner, as shown in the proof of Lemma \ref{disj}. These are at most $2\ell \beta + 2\ell^2+3\ell$ edges not in $D$, which is at most $c\ell^2$ since $\beta \leq \ell$ and $c=7$.
This is also a $k$-spanner for any $k \geq 5$.

If the input strings are far from being disjoint then there are at least $\frac{\ell^2}{12}$ pairs $(i,r)$ such that none of the edges $(x^1_i,x^2_r),(y^1_i,y^2_r)$ are in $G$.
Hence, by Claim \ref{directed}, there are at least $\frac{\ell^2}{12}$ pairs $(i,r)$ such that there is no directed path between the vertices $x_{ij},y_{rs}$ except for the path that consists of the edge $(x_{ij},y_{rs})$. For each such pair, we need to take all the directed edges $(x_{ij},y_{rs})$ to the spanner for all the values of $j$ and $s$, which means adding $\beta^2$ edges to the spanner. Summing over all the $\frac{\ell^2}{12}$ pairs, we get that any $k$-spanner must include at least $\frac{\beta^2}{12} \ell^2$ edges of $D$.
\end{proof}

Let $k \geq 5$ and let $A$ be a deterministic distributed $\alpha$-approximation algorithm for the minimum $k$-spanner problem. Denote by $T(n)$ the round complexity of $A$ on a graph with $n$ vertices. The following lemma adapts Lemma \ref{threshold} to the gap disjointness problem. Its proof is the same as the proof of Lemma \ref{threshold}, with the difference that now Alice concludes that the input strings are \textit{far from being disjoint} if and only if the constructed spanner has more than $\alpha(n) \cdot t$ edges of $D$. Also, now the lower bound holds only for the deterministic case, since it relies on the communication complexity of gap disjointness.

\begin{lemma} \label{threshold2}
Let $G=G(\ell,\beta)$. If there is a threshold $t$ such that if the input strings $a,b$ are disjoint then an optimal $k$-spanner of $G$ has at most $t$ edges, and if the input strings are far from being disjoint then each $k$-spanner of $G$ includes more than $\alpha(n) \cdot t$ edges of $D$. Then, $T(n) = \Omega(\frac{\ell}{\log{n}}).$
\end{lemma}

Using Lemma \ref{disj2} and Lemma \ref{threshold2}, we show the following.

\begin{theorem} \label{deterministic}
Any deterministic distributed $\alpha$-approximation algorithm in the \congest model for the directed $k$-spanner problem for $k \geq 5$ takes $\Omega(\frac{n}{\sqrt{\alpha} \cdot \log{n}})$ rounds, for $1 \leq \alpha \leq \frac{n}{c'}$ for a constant $c' > 1$.
\end{theorem}

\begin{proof}
We construct the graph $G=G(\ell,\beta)$ with the following choice for the parameters $\ell, \beta$. Let $n'$ be a positive integer, and let $c=7$. Let $\beta = \lceil \sqrt{12\alpha(n') c} \rceil +1$, and let $\ell = \lfloor \frac{n'}{c \beta} \rfloor$.

The number of vertices in $G$ is $n = \Theta( \ell \beta) = \Theta(\frac{n'}{\beta} \beta) =\Theta(n').$ In addition, it holds that $n \leq c \ell \beta$ since the number of vertices in $G$ is $2 \ell \beta + 5 \ell$, which gives
$n \leq c \ell \beta \leq \frac{n'}{c \beta} c \beta = n'.$
In order to use Lemma \ref{disj2} we need to verify that $\beta \leq \ell.$ Note that $n = c_1 \ell \beta$ for a constant $2 \leq c_1 \leq c$. It follows that $\beta \leq \ell$ if and only if $c_1 \beta^2 \leq n$.
Since $\beta = \Theta(\sqrt{\alpha(n)}) = c_2 \sqrt{\alpha(n)}$ for a constant $c_2$, if we choose $c' = c_1 c_2^2$, we get that if $\alpha(n) \leq \frac{n}{c'}$, then $\beta \leq \sqrt{\frac{n}{c_1}}$, which gives $c_1 \beta^2 \leq n$ as needed.

We now define $t = c\ell^2$. By Lemma \ref{disj2}, if the input strings $a,b$ are disjoint, then there is a $k$-spanner of size at most $t = c \ell^2$. Otherwise, if the input strings are far from being disjoint, then any $k$-spanner for $G$ includes at least $\frac{\beta^2}{12} \ell^2$ edges of $D$. By the choice of $\beta$ and since $n \leq n'$, it holds that $12 \alpha(n) \cdot c < \beta^2$, which gives $\alpha(n) \cdot t = \alpha(n) \cdot c \ell^2 < \frac{\beta^2}{12} \ell^2$, which shows that $t$ satisfies the conditions of Lemma \ref{threshold2}.
Using Lemma \ref{threshold2} we get that $T(n) = \Omega(\frac{\ell}{\log{n}}).$
Note that now $\ell = \Theta(\frac{n}{\beta}) = \Theta(\frac{n}{\sqrt{\alpha(n)}})$, which shows that $T(n) = \Omega(\frac{n}{\sqrt{\alpha(n)} \cdot \log{n}}).$
\end{proof}

Theorem \ref{deterministic} shows that achieving a constant or a polylogarithmic approximation ratio for the directed $k$-spanner problem in the \congest model requires $\widetilde{\Omega}(n)$ rounds for any deterministic algorithm. In addition, even an approximation ratio of $n^{\epsilon}$ is hard, requiring $\widetilde{\Omega}(n^{1-\epsilon/2})$ rounds, for any $0 < \epsilon < 1$. Notably, even an approximation ratio of $\frac{n}{c}$ for appropriate values of $c$ is hard, requiring $\widetilde{\Omega}(\sqrt{n})$ rounds. This is to be contrasted with the fact that obtaining an approximation ratio of $n$ requires no communication, since any $k$-spanner has at least $n-1$ edges.

Theorem \ref{deterministic} separates the \local and the \congest models, since the deterministic network decomposition described in \cite{barenboim2016fast} gives a deterministic $O(n^{\epsilon})$-approximation for directed $k$-spanner for a constant $k$ in polylogarithmic time in the \local model.

It also separates the undirected and directed $k$-spanner problems for deterministic algorithms. Currently the best deterministic algorithm in the \congest model for the undirected $k$-spanner problem, is a recent algorithm \cite{GrossmanParter} which constructs $(2k-1)$-spanners of size $O(n^{1+1/k})$ in $O(n^{1/2-1/k})$ rounds for a constant even $k$ (in the \local nodel there is a $k$-round deterministic algorithm for this problem \cite{derbel2008locality}). This gives an $O(n^{1/k})$-approximation for undirected $(2k-1)$-spanners. Achieving the same approximation for the directed problem requires $\widetilde{\Omega}(n^{1-1/{2k}})$ rounds according to Theorem \ref{deterministic}.

\subsection{Weighted $k$-spanner} \label{sec:weighted}

We extend our construction to the weighted case, showing that any approximation for the weighted $k$-spanner in the \congest model takes $\widetilde{\Omega}(n)$ rounds for $k \geq 4$, even for randomized algorithms. A similar result holds for the weighted undirected case. In the weighted case, rather than guaranteeing that each input bit affects many edges of the spanner, we simply assign weight 0 to all the edges that are not in $D$ and weight 1 to all the edges of $D$. Hence, taking even a single edge from $D$ is very expensive if we can avoid it. This allows us to show a simpler construction, obtaining a stronger lower bound for the weighted case, as follows.

We build a graph $G_w(\ell) = G_w = (V_w, E_w)$ which is the same as $G$, except for the following differences (see Figure \ref{alice_bob_w}). We define $\beta = 1$, and change the set of vertices to be $V_w = V \setminus Y_3$. Since $\beta =1$, the vertices in $X_2$ and $Y_2$ are only of the form $x_{i1},y_{i1}$ for $1 \leq i \leq \ell$. We change their names from $x_{i1},y_{i1}$ to $x_i$ and $y_i$, respectively.
For each $1 \leq i \leq \ell$ we replace the two edges $(y^2_i,y^3_i),(y^3_i,y_{i1}) \in E$ by the edge $(y^2_i,y_{i}) \in E_w.$
Since $\beta =1$, the size of the cut between $Y_1$ and the rest of the graph is still $\Theta(\ell)$.

\setlength{\intextsep}{0pt}
\begin{figure}[h]
\centering
\setlength{\abovecaptionskip}{-2pt}
\setlength{\belowcaptionskip}{8pt}
\includegraphics[scale=0.5]{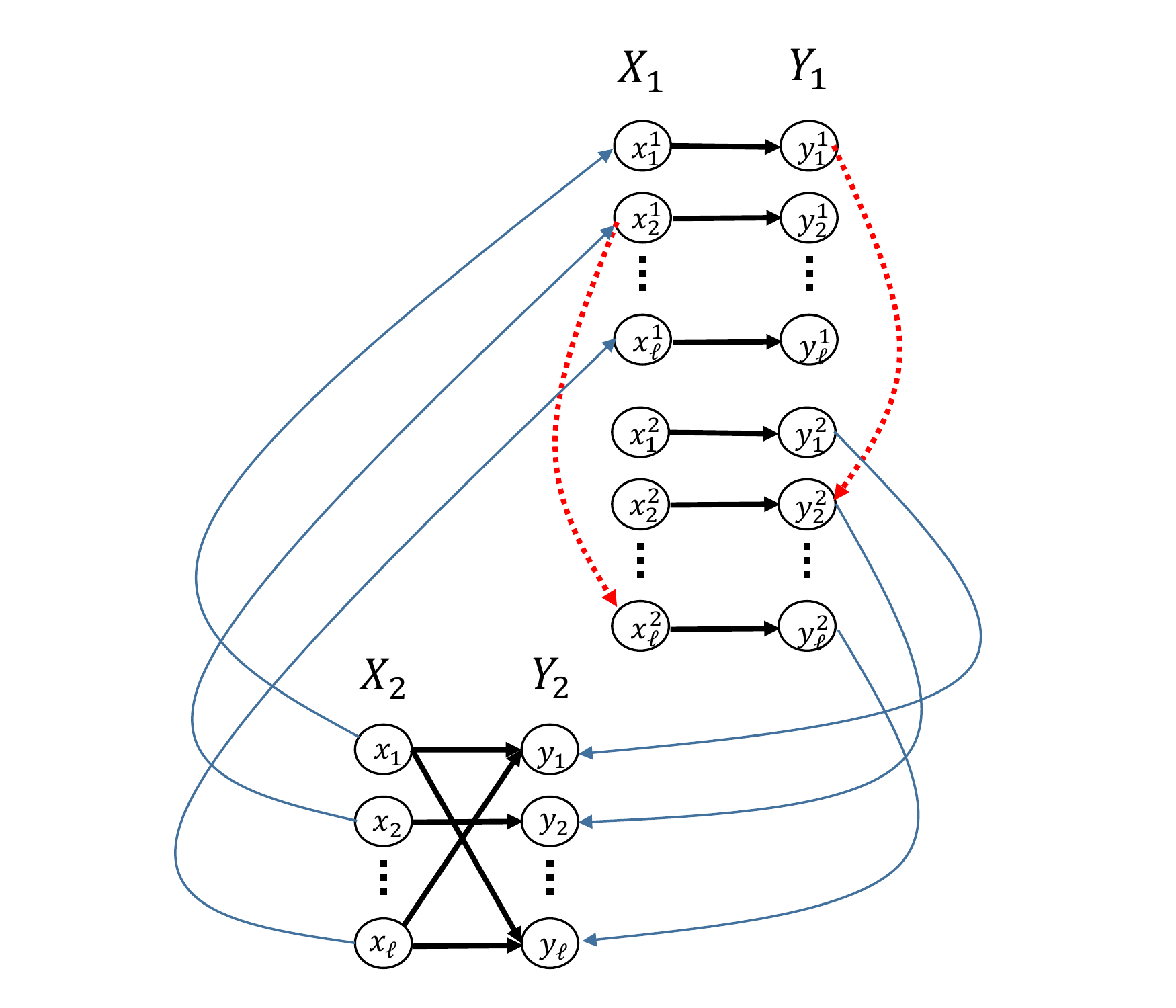}
 \caption{The graph $G_w$, with some of its edges omitted for clarity. The red dashed edges are examples of optional edges which depend on the input strings $a$ and $b$.}
\label{alice_bob_w}
\end{figure}

The following theorem states our lower bound for the weighted directed case.

\begin{theorem}
Any (perhaps randomized) distributed $\alpha$-approximation algorithm in the \congest model for the weighted directed $k$-spanner problem for $k \geq 4$ takes $\Omega(\frac{n}{\log{n}})$ rounds.
\end{theorem}

\begin{proof}
Let $n= 6\ell$ for a positive integer $\ell$, and let $G_w = G_w(\ell)$. Note that the number of vertices in $G_w$ is exactly $n$.
There is a 4-spanner of cost $0$ for $G_w$ if and only if there is a path of length at most $4$ of edges of weight $0$ between every pair of vertices $x_{i},y_{j}$. A path of length at most $4$ between $x_{i}$ and $y_{j}$ that includes only edges of weight 0, must start with the edge $(x_{i},x^1_i)$ and must end with the edge $(y^2_j,y_{j})$. Following the proof of Claim \ref{directed}, we argue that there is such a path if and only if one of the edges $(x^1_i,x^2_j),(y^1_i,y^2_j)$ is in $G_w$. Otherwise, there is no directed path of weight 0 between $x_{i}$ and $y_{j}$.

It follows that for every $k \geq 4$, there is a $k$-spanner of cost 0 for $G_w$ if and only if the inputs $a$ and $b$ are disjoint. Hence, a distributed $\alpha$-approximation algorithm $A$ for the weighted $k$-spanner problem can be used to solve set disjointness: we define $V_B = Y_1, V_A = V \setminus V_B$ and let Alice and Bob simulate the algorithm on $G_w$ as before. At the end of the simulation, Alice concludes that the inputs are disjoint if and only if none of the edges of $D$ are taken to the spanner.

If the inputs are disjoint, then there is a spanner of cost 0. Hence, for any $\alpha \geq 1$, an $\alpha$-approximation must return a spanner of cost 0 if such exists. Otherwise, any spanner must include at least one of the edges of $D$ which proves that the output of Alice is indeed correct.

As in the proof of Lemma \ref{threshold}, we get that $T(n) = \Omega(\frac{\ell}{\log{n}})$. Since $n = 6 \ell$, this gives $T(n) = \Omega(\frac{n}{\log{n}})$.
\end{proof}

We prove a similar bound for the weighted undirected $k$-spanner problem for $k \geq 4$.
In the undirected case, we would like to construct a similar graph $G_w$, with only modifying all of its edges to be undirected. It would still hold that there is a path of length at most $4$ of edges of weight $0$ between the vertices $x_{i},y_{j}$ if and only if one of the edges $\{x^1_i,x^2_j\},\{y^1_i,y^2_j\}$ is in $G_w$, following the same proof. However, since the edges are undirected, there may be a path of length longer than $4$ of edges of weight 0 between the vertices $x_{i},y_{j}$, even if none of the edges $\{x^1_i,x^2_j\},\{y^1_i,y^2_j\}$ is in $G_w$, which requires us to modify our construction in order for our bounds to apply also for $k>4$.

We change the construction of $G_w$ as follows. For each $1 \leq i \leq \ell$ we replace the edge $\{y^2_i, y_{i} \}$ by a path of length $k-3$, by adding to the graph $k-4$ vertices $y^3_i,...,y^{k-2}_i$ and the required edges for constructing the path $\{y^2_i,y^3_i,...,y^{k-2}_i,y_{i}\}$. All of the edges of the path have weight 0.

Any path of length at most $k$ of edges of weight 0 between $x_{i}$ to $y_{j}$ must start with the edge $\{x_{i},x^1_i\}$ and must end with the path of length $k-3$ that we added between $y^2_j$ to $y_{j}$. Hence, there is a path of length at most $k$ between $x_{i}$ to $y_{j}$ of edges of weight 0 if and only if there is a path of length 2 between $x^1_i$ and $y^2_j$. This can only happen if one of the edges $\{x^1_i,x^2_j\},\{y^1_i,y^2_j\}$ is in the graph.

The rest of the proof is exactly the same as in the directed case. However, we added $\Theta(k\ell)$ vertices to the graph. Hence, the number of vertices in the graph is $n = \Theta(k \ell)$ and not $\Theta(\ell)$ as before, which gives $\ell = \Theta(\frac{n}{k})$. This allows us to prove a lower bound of $\widetilde{\Omega}(\frac{n}{k})$ to the undirected problem (which is still $\widetilde{\Omega}(n)$ for small values of $k$).

\begin{theorem}
Any (perhaps randomized) distributed $\alpha$-approximation algorithm in the \congest model for the weighted undirected $k$-spanner problem for $k \geq 4$ takes $\Omega(\frac{n}{k \cdot \log{n}})$ rounds.
\end{theorem}

\section{Hardness of approximation of weighted 2-spanner} 
\label{sec:MVC}

In this section, we show that approximating the weighted 2-spanner problem is at least as hard as approximating the (unweighted) minimum vertex cover (MVC) problem. Therefore, known lower bounds for MVC translate directly to lower bounds for the weighted 2-spanner problem.

In the MVC problem the input is a graph $G=(V,E)$ and the goal is to find a minimum set of vertices $C$ that covers all the edges. That is, it is required that for each edge $e=\{v,u\}$, at least one of $v$ and $u$ is in $C$.

Let $G=(V,E)$ be an input graph to the MVC problem. We construct a new graph $G_S=(V_S,E_S)$ in the following way (see Figure \ref{vc}). For each vertex $v \in V$, there are 3 vertices in $V_S$: $v_1,v_2,v_3$. We connect these 3 vertices by a triangle, where the edge $\{v_1,v_2\}$ has weight 1, and the edges $\{v_1,v_3\},\{v_2,v_3\}$ have weight 0.
In addition, for each edge $\{v,u\} \in E$, there are 3 edges in $E_S$: $\{v_1,u_1\}, \{v_2,u_2\}$, both having weight 0, and one of the edges $\{v_1,u_2\},\{u_1,v_2\}$, according to the order of the IDs of $v$ and $u$, having weight 2.

\setlength{\intextsep}{2pt}
\begin{figure}[h]
\centering
\setlength{\abovecaptionskip}{-6pt}
\setlength{\belowcaptionskip}{8pt}
\includegraphics[scale=0.5]{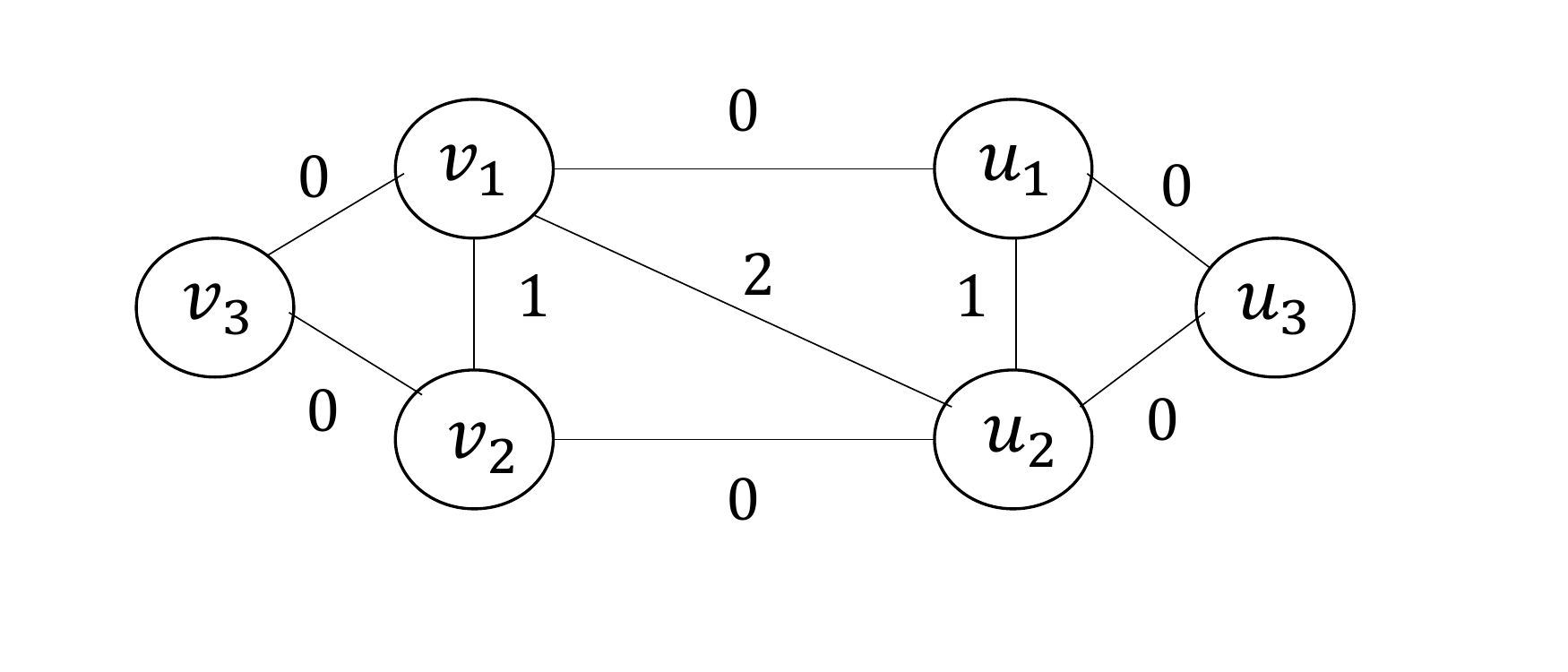}
\caption{For each vertex $v \in G$, there is a corresponding triangle in $G_S$ between the vertices $v_1,v_2,v_3$. The edge $\{v,u\} \in G$ has 3 corresponding edges in $G_S$: $\{v_1,u_1\}, \{v_2,u_2\}$, $\{v_1,u_2\}$.} \label{vc}
\end{figure}

We show that a solution for the weighted 2-spanner problem in $G_S$ gives a solution for MVC in $G$.

\begin{claim} \label{reduction}
The cost of the minimum 2-spanner in $G_S$ is exactly the size of the minimum vertex cover in $G$.
\end{claim}

\begin{proof}
Let $C$ be a minimum vertex cover in $G$. We construct a 2-spanner $H_C$ for $G_S$ as follows. First, $H_C$ includes all the edges having weight 0. In addition, for every $v \in C$, we add to $H_C$ the edge $\{v_1,v_2\}$. Note that these edges have weight 1, and all the other edges we add to $H_C$ have weight 0, hence, the cost of $H_C$ is exactly $|C|$. We now show that $H_C$ is a 2-spanner. All the edges having weight 0 in $G_S$ are added to the spanner, and hence they are covered. All the edges having weight 1 are covered by edges of weight 0, since an edge $\{v_1,v_2\}$ is covered by the path $\{v_1,v_3\},\{v_3,v_2\}$. Let $\{v_1,u_2\}$ be an edge of weight 2 in $G_S$, and let $\{v,u\}$ be the corresponding edge in $G$. Since $C$ is a vertex cover, at least one of the vertices $v,u$ is in $C$. In the former case, we add $\{v_1,v_2\}$ to $H_C$, hence, the edge $\{v_1,u_2\}$ is covered by the path $\{v_1,v_2\}, \{v_2,u_2\}$ (note that $\{v_2,u_2\}$ has weight 0 and is included in $H_C$). In the latter case, $\{v_1,u_2\}$ is covered by the path $\{v_1,u_1\},\{u_1,u_2\}$. Hence, $H_C$ is a 2-spanner having cost $|C|$.

In the other direction, let $H$ be a minimum cost 2-spanner in $G_S$ having cost $w(H)$. We construct a vertex cover $C_H$ in $G$ with size $w(H)$. We start by converting $H$ into a 2-spanner $H'$ with the same cost. First, $H'$ contains all the edges having weight 0 in $G_S$ and all the edges having weight 1 in $H$. In addition, if $H$ includes an edge $\{v_1,u_2\}$ having weight $2$, we replace it in $H'$ by the two edges $\{v_1,v_2\},\{u_1,u_2\}$, each having weight 1. This transformation clearly cannot increase the cost. We next show that $H'$ is still a 2-spanner. Since $H'$ includes all the edges having weight 0 in $G_S$, it covers all the edges of weight 0 or 1 by edges of weight 0, as explained above. In addition, any edge of weight 2 that is covered in $H$ by a path of length 2 that includes only edges of weight 0 or 1, is covered in $H'$ in the same way.
Let $e = \{x_1,y_2\}$ be an edge of weight 2, covered in $H$ by a path $P$ of length at most 2 that includes the edge $e' = \{v_1,u_2\} \in H$ having weight 2 ($e'$ may be different than $e$). It holds that $e' \not \in H'$, since $e'$ has weight 2, hence, we added the edges $\{v_1,v_2\}$ and $\{u_1,u_2\}$ to $H'$. Since $P$ has length at most 2, it follows that $x_1 = v_1$ or $y_2 = u_2$. In the first case, $e=\{v_1,y_2\}$ and the path $\{v_1,v_2\},\{v_2,y_2\}$ covers $e$ (we added $\{v_1,v_2\}$ to $H'$, and since $\{v_1,y_2\} \in E_S$, then $\{v_2,y_2\}$ is also in $E_S$ and has weight 0). In the second case, $e=\{x_1,u_2\}$ and $\{x_1,u_1\},\{u_1,u_2\}$ is a path of length 2 that covers $e$ in $H'$.

Therefore, $H'$ is a 2-spanner with the same cost of $H$. We define $C_H = \{v| \{v_1,v_2\} \in H' \}$. The size of $C_H$ is exactly $w(H')$ since the edges $\{v_1,v_2\} \in H'$ are exactly all the edges of $H'$ having weight 1, and $H'$ includes only edges of weight 0 or 1. In addition, we claim that $C_H$ is a vertex cover. Let $\{v,u\} \in E$, then one of the edges $\{v_1,u_2\}$ or $\{u_1,v_2\}$ is in $E_S$. Assume w.l.o.g that $e=\{v_1,u_2\} \in E_S$. Note that $e \not \in H'$ since it has weight 2. Since $H'$ is a 2-spanner it includes a path of the form $\{v_1,x\},\{x,u_2\}$ that covers $e$. It must hold that $x = v_2$ or $x = u_1$ (if, for example, $x = w_1$ such that $w_1 \neq u_1$, then the edge $\{w_1,u_2\}$ has weight 2 and is not in $H'$). Hence, at least one of the edges $\{v_1,v_2\},\{u_1,u_2\}$ is in $H'$, which means that at least one of $v,u$ is in $C_H$ as needed.

In conclusion, the cost of a minimum 2-spanner in $G_S$ is exactly the size of the minimum vertex cover in $G$.
\end{proof}

We can now relate the above to the number of rounds required for distributed algorithms that solve or approximate these two problems.

\begin{lemma} \label{reduction2}
Let $A$ be a distributed $\alpha$-approximation algorithm for the weighted 2-spanner problem that takes $T(n)$ rounds on a graph with $n$ vertices. Then there is an $\alpha$-approximation algorithm for MVC that takes $3T(3n)$ rounds on a graph with $n$ vertices.
\end{lemma}

\begin{proof}
We describe an algorithm $A_{MVC}$ that approximates MVC. Let $G$ be an input graph for MVC. The algorithm $A_{MVC}$ simulates $A$ on the graph $G_S$, in the following way. Each vertex $v \in V$ simulates $A_{MVC}$ on the vertices $v_1,v_2,v_3$. Each time a message is sent on one of the 3 edges corresponding to an edge $\{v,u\} \in E$, we send this message over the edge $\{v,u\} \in E$. Since we may need to send 3 different messages on this edge, each round of $A$ can be simulated in three rounds of $A_{MVC}$.\footnote{In the \local model we can send these 3 messages in one round. However, we spend three different rounds in order for the simulation to work also in the \congest model.} When $A$ finishes, we convert the solution $H$ to a vertex cover $C_H$ as described in the proof of Claim \ref{reduction}, without any communication. From Claim \ref{reduction}, it follows that if $H$ is an $\alpha$-approximation for the weighted 2-spanner problem in $G_S$, then $C_H$ is an $\alpha$-approximation for MVC in $G$. Let $n$ be the number of vertices in $G$. The number of vertices in $G_S$ is $3n$ from the definition of $G_S$, hence the time complexity of simulating $A$ on $G_S$ is $3T(3n)$.
\end{proof}

Lemma~\ref{reduction2} shows that if $A$ works in the \congest model, then $A_{MVC}$ works in the \congest model as well. Hence, lower bounds for approximating MVC in both the \local and the \congest models give lower bounds for the weighted 2-spanner problem. This gives the following results.

\begin{theorem} \label{local}
To obtain a constant or a polylogarithmic approximation ratio for the weighted 2-spanner problem, even in the \local model, there are graphs on which every distributed algorithm requires at least $\Omega(\frac{\log{\Delta}}{\log{\log{\Delta}}})$ rounds and $\Omega(\sqrt{\frac{\log{n}}{\log{\log{n}}}})$ rounds.
\end{theorem}

Theorem \ref{local} follows from Theorem 14 in \cite{kuhn2016local} and from Lemma \ref{reduction2}. Note that the number of vertices and the maximum degree in $G_S$ are equal up to a constant factor to the number of vertices and maximum degree in $G$.
In addition, Theorem 13 in \cite{kuhn2016local}, allows us to show trade-offs between the time complexity of a distributed algorithm for weighted 2-spanner to the approximation ratio it gets.

\begin{theorem} \label{local2}
For every integer $k > 0$, there are graphs $G$, such that in $k$ communication rounds in the \local model, every distributed algorithm for the weighted 2-spanner problem on $G$ has approximation ratios of at least $\Omega\left( \frac{n^\frac{1-o(1)}{4k^2}}{k} \right) $ and $ \Omega\left( \frac{\Delta^\frac{1}{k+1}}{k} \right) $.
\end{theorem}

In the \congest model, solving MVC optimally takes $\Omega\left( \frac{n^2}{\log^2{n}}\right) $ rounds (see Theorem 2 in \cite{censor2017quadratic}), which carries over to exact spanners, as follows.

\begin{theorem}
Any distributed algorithm in the \congest model that solves the weighted 2-spanner problem optimally requires $\Omega\left( \frac{n^2}{\log^2{n}}\right) $ rounds.
\end{theorem}

All of these lower bounds hold also for randomized algorithms.

\paragraph*{Remarks:} Our reduction from MVC can be adapted to obtain additional bounds. First, by changing the weights of all edges having weight 2 to have weight 1, we obtain that an $\alpha$-approximation for the weighted 2-spanner problem gives a $2\alpha$-approximation for MVC. This implies the lower bounds of Theorem~\ref{local} and Theorem~\ref{local2} also for graphs with $0,1$ weights. This can be viewed as lower bounds for the 2-spanner augmentation problem, in which we are given an initial set of edges and need to augment it with the minimal number of edges that induces a 2-spanner.

Further, the same lower bounds hold for the \emph{directed} weighted case. We modify the construction such that the edges of the triangle for vertex $v$ are $(v_1,v_2),(v_1,v_3),(v_3,v_2)$. For an edge $\{v,u\} \in E$, $E_S$ includes $5$ directed edges: $(v_1,u_1),(u_1,v_1),(v_2,u_2),(u_2,v_2)$ and one of the edges $(v_1,u_2),(u_1,v_2)$. The weights of all the edges remain as in the undirected case.

\section{Distributed approximation for 2-spanner problems} \label{sec:alg}


Here we present our distributed approximation algorithm for the minimum 2-spanner problem. We need the following terminology and notation.

A \textit{$v$-star} is a non-empty subset of edges between $v$ and a subset of its neighbors.
The \textit{density} of a star $S$ with respect to a subset of edges $H$, denoted by $\rho(S,H)$, equals $\frac{|C_S|}{|S|}$, where $C_S$ is the set of edges of $H$ \emph{2-spanned} by the star $S$, where an edge $e=\{u,w\}$ is 2-spanned by the $v$-star $S$ if $S$ includes the edges $\{v,u\},\{v,w\}$. Note that $S$ \emph{covers} all the edges 2-spanned by $S$ and also the edges of $S$, but it 2-spans only non-star edges.
The \textit{densest} $v$-star with respect to $H$ is the $v$-star having maximal density with respect to $H$. The density of a vertex $v$ with respect to $H$, denoted by $\rho(v,H)$, is the density of the densest $v$-star. If $H$ is clear from the context, we refer to $\rho(S,H)$ and $\rho(v,H)$ as the density of $S$ and the density of $v$, and denote them by $\rho(S)$ and $\rho(v)$, respectively. The \textit{rounded} density of a star $S$ with respect to $H$, denoted by $\tilde{\rho}(S,H)$, is obtained by rounding $\rho(S,H)$ to the closest power of 2 that is greater than $\rho(S,H)$. Similarly, the \textit{rounded} density of a vertex $v$ with respect to $H$, denoted by $\tilde{\rho}(v,H)$, is obtained by rounding $\rho(v,H)$ to the closest power of 2 that is greater than $\rho(v,H)$.
The \textit{full} $v$-star is the star that includes all the edges between $v$ and its neighbors.
The \textit{2-neighborhood} of a vertex $v$ consists of all the vertices at distance at most $2$ from $v$.

In our algorithm, each vertex $v$ maintains a set $H_v$ that includes all the edges 2-spanned by the full $v$-star that are still not covered by the edges added to the spanner.
The algorithm proceeds in iterations, where in each iteration the following is computed:\\

\noindent\fbox{%
    \parbox{\textwidth}{%
    \vspace{-0.3cm}
\begin{enumerate}[rightmargin=0.2cm]
\item { Each vertex $v$ computes its rounded density $\rho_v = \tilde{\rho}(v, H_v)$, and sends it to its 2-neighborhood.}
\item \label{choose_star} { Each vertex $v$ such that $\rho_v \geq \rho_u$ for each $u$ in its 2-neighborhood and $\rho(v,H_v) \geq 1$ is a \textit{candidate}. Let $S_v$ be a $v$-star with density at least $\frac{\rho_v}{4}$, chosen according to Section \ref{choose} (a choice which is central for our analysis to carry through). Vertex $v$ informs its neighbors about $S_v$. Let $C_v$ be the edges of $H_v$ 2-spanned by $S_v$.}
\item { Each candidate $v$ chooses a random number $r_v \in \{1,...,n^4\}$ and sends it to its neighbors.\footnotemark }
\item { Each uncovered edge that is 2-spanned by at least one of the candidates, votes for the first candidate that 2-spans it according to the order of the values $r_v$. If there is more than one candidate with the same minimum value, it votes for the one with minimum ID.}
\item { Each star $S_v$ for which $v$ receives at least $\frac{|C_v|}{8}$ votes from edges it 2-spans is added to the spanner.}
\item { Each vertex $v$ updates the set $H_v$ in its 2-neighborhood by removing from it edges that are now covered. }
\item{If the maximal density in the 2-neighborhood of $v$ is at most $1$, $v$ adds to the spanner all the edges adjacent to it that are still not covered, and outputs the edges adjacent to it that were added to the spanner during the algorithm.}
\end{enumerate}
    \vspace{-0.4cm}
}}

\footnotetext{Knowing the exact value of $n$ is unnecessary, and the typical assumption that the vertices know a polynomial upper bound on $n$ suffices.}

~\\ 

At the end of the algorithm all the edges are covered by spanner edges, since we add to the spanner edges that are not 2-spanned during the algorithm.

Since all the candidates have maximal rounded density in their 2-neighborhood, it follows that
all the candidates that cover the same edge have the same rounded density, which is crucial in the analysis. In addition, rounding the densities guarantees that there are only $O(\log{\Delta})$ possible values for the maximal rounded density, which allows us to show an efficient time complexity.

Each iteration takes constant number of rounds in the \local model. For example, to calculate $\tilde{\rho}(v, H_v)$, each vertex $v$ learns all the edges between its neighbors that are still uncovered, by having each vertex $u$ send to its neighbors a list of its neighbors $w$ such that the edges $\{u,w\}$ are still not covered. We next show that the algorithm requires only polynomial local computations.

We can compute the densest $v$-star in polynomial time as in the sequential algorithm (see Lemma 2.1 in \cite{kortsarz1994generating}). This is the maximal density problem, that can be solved in polynomial time using flow techniques \cite{gallo1989fast}. 
This allows us to compute the rounded density of a vertex. We next explain how we choose the star $S_v$ in polynomial time. Other computations in the algorithm are clearly polynomial. 

\subsection{Choosing the star $S_v$} \label{choose}
In step \ref{choose_star} of each iteration, a candidate vertex $v$ chooses a $v$-star with density at least $\frac{\rho_v}{4}$. However, there may be multiple $v$-stars with such density, and choosing an arbitrary star between them \emph{does not} meet our claimed round complexity, and it is crucial to choose the stars in a certain way. In addition, we have to find such star using only polynomial local computations.
We next describe how to choose the star $S_v$.

Let $H_v^i$ be the subset $H_v$ at the beginning of iteration $i$. It holds that $H_v^{i+1} \subseteq H_v^{i}$ for all $i$. Let $\rho = \tilde{\rho}(v,H_v^i)$. The star $S_v^i$ that $v$ chooses in iteration $i$ is defined as follows.
If $i$ is the first iteration in which $v$ is a candidate with rounded density $\rho$, then $S_v^i$ is chosen as follows. First, $v$ computes the densest $v$-star, denote it by $S$. 
Now, if there is an edge $e$ such that $\rho(S \cup \{e\}, H_v^i) \geq \frac{\rho}{4}$, then $v$ adds such an edge to $S$. Otherwise, if there is a disjoint $v$-star $S'$ such that $\rho(S',H_v^i) \geq \frac{\rho}{4}$, then $v$ adds the edges of $S'$ to $S$. Now $v$ continues in the same manner until there is no edge or disjoint star it can add to $S$ without decreasing the density below $\frac{\rho}{4}$. The resulting star is $S_v^i$. 

If $v$ is already a candidate with rounded density $\rho$ in iteration $i-1$, then if $\rho(S_v^{i-1},H_v^i) \geq \frac{\rho}{4}$, we define $S_v^i = S_v^{i-1}$. Otherwise, if $S_v^{i-1}$ contains a star with density at least $\frac{\rho}{4}$ with respect to $H_v^i$, we define $S_v^i$ as follows. $v$ starts by computing the densest $v$-star $S$ that is contained in $S_v^i$, and then adds to it edges or disjoint stars as before, however, it only considers adding edges or disjoint stars from $S_v^{i-1}$. This guarantees that $S_v^i \subseteq S_v^{i-1}$.
If $S_v^{i-1}$ does not contain a star with density at least $\frac{\rho}{4}$ with respect to $H_v^i$, $v$ chooses an arbitrary $v$-star with rounded density $\rho$. (We later show that this never happens).

The computations are polynomial. $v$ adds edges to $S$ at most $n$ times. Each time it adds edges it does the following computation: it checks if there is an edge such that $\rho(S \cup \{e\}, H_v^i) \geq \frac{\rho}{4}$, and since there are at most $n$ optional edges, the computation is polynomial. It also checks if there is a disjoint star with density at least $\frac{\rho}{4}$. To compute this, it computes the densest star that is disjoint to $S$.

\subsection{Analysis} \label{sec:analysis}

In this section, we present the analysis of our distributed approximation algorithm for the minimum 2-spanner problem, 
and prove the following theorem.

\minimumSpanner*

Let $H$ be the set of edges of the spanner produced by the algorithm. When the algorithm ends, all the edges are covered, hence $H$ is a 2-spanner. We show that its size it at most $O(\log{\frac{m}{n}})|H^*|$, where $H^*$ is the set of edges of a minimum 2-spanner. Afterwards, we show that the time complexity of the algorithm is $O(\log{n} \log{\Delta})$ rounds, w.h.p.

\subsubsection{Approximation ratio} \label{sec:approx}

We start by showing that our algorithm guarantees an approximation ratio of $O(\log{\frac{m}{n}})$.
The analysis of the sequential algorithm of \cite{kortsarz1994generating} that obtains the same approximation ratio strongly depends on the facts that stars are added to the spanner one by one and that the star that is added at each step has maximal rounded density. These allow dividing the edges to several subsets according to the order in which they are covered in the algorithm, and bounding the number of edges in each subset.

Our analysis borrows ideas from the above analysis, but requires a more sophisticated accounting, since our algorithm adds multiple stars in each iteration, with varying densities. In addition to overcoming these uncertainties, a compelling aspect of our approach is that it easily extends to other variants of the problem, such as the client-server 2-spanner problem \cite{elkin1999client}.

To show the approximation ratio, we assign each edge $e \in E$ a value $cost(e)$ such that the sum of the costs of all edges is closely related both to $|H|$ and $|H^*|$, by satisfying $$|H| \leq 8 \sum_{e \in E} cost(e) \leq O\left( \log{\frac{m}{n}}\right) |H^*|,$$
which implies our claimed approximation ratio.

We write $H=H_1 \cup H_2$, where $H_1$ are edges added to the spanner during the algorithm, and $H_2$ are edges added to the spanner at the end of the algorithm, when the maximal density in the 2-neighborhood of a vertex is at most 1.

For an edge $e \in H_2$, we set $cost(e) = 1$.
For an edge $e \in H_1$, let $i$ be the iteration in which $e$ is first covered in the algorithm.
The edge $e$ may be covered by a candidate star $S_v$ that it votes for and is added to the spanner at iteration $i$.  In this case, we set $cost(e) = \frac{1}{\rho}$, where $\rho$ is the density of the star $S_v$ that $e$ chooses at iteration $i$. Another option is that $e$ is covered as a result of adding other stars to the spanner at iteration $i$: it may be covered either by a different star than the one it votes for, or by a path of length 2 that is created by edges added to the spanner at iteration $i$ together with edges added at previous iterations. In each of these cases, we set $cost(e) = 0$.
We first show the left inequality above.

\begin{lemma} \label{cost}
$|H| \leq 8 \cdot \sum_{e \in E} cost(e)$.
\end{lemma}

\begin{proof}
To prove that $|H| \leq 8\sum_{e \in E} cost(e)$, it is enough to show that $|H_1| \leq 8\sum_{e \in E \setminus H_2} cost(e)$ and $|H_2| \leq 8\sum_{e \in H_2} cost(e)$. The second inequality follows since $|H_2| = \sum_{e \in H_2} 1 = \sum_{e \in H_2} cost(e).$ We next prove the first inequality.

Let $Stars$ be the set of stars added to $H_1$ in the algorithm. It holds that $|H_1| \leq \sum_{S \in Stars} |S|$, since each edge of $H_1$ is included in at least one star. Let $S_v$ be a star added to $H_1$ at iteration $i$, having density $\rho$ at that iteration. Recall that we add $S_v$ to the spanner since it gets at least $\frac{|C_v|}{8}$ votes from the edges it 2-spans. Denote by $Votes(S_v)$ the set of edges that vote for $S_v$ at iteration $i$. For each $e \in Votes(S_v)$, we defined $cost(e) = \frac{1}{\rho}$, which gives, $$\sum_{e \in Votes(S_v)} cost(e) \geq \frac{1}{\rho} \cdot \frac{|C_v|}{8} =  \frac{|S_v|}{|C_v|} \cdot \frac{|C_v|}{8} = \frac{|S_v|}{8}.$$
Hence, for each $S \in Stars$, $|S| \leq 8 \cdot \sum_{e \in Votes(S)} cost(e).$
For each edge $e$, there is at most one star $S \in Stars$ such that $e \in Votes(S)$, since an edge votes for at most one star at the iteration in which it is covered. In addition, an edge $e \in Votes(S)$ is 2-spanned during the algorithm, which means that $e \not \in H_2$. Hence, we get $$|H_1| \leq \sum_{S \in Stars} |S| \leq 8 \cdot \sum_{S \in Stars} \sum_{e \in Votes(S)} cost(e) \leq 8 \cdot \sum_{e \in E \setminus H_2} cost(e).$$
This completes the proof of Lemma \ref{cost}.
\end{proof}

To bound $\sum_{e \in E} cost(e)$ from above, let $r = \frac{m}{n}$, and $f=\lceil \log{r} \rceil$.
We divide the edges of $E$ to $f+2$ subsets $\{E_j\}_{j=0}^{f+1}$ according to their costs, and show that for each $j$, the sum of costs of edges in $E_j$ is at most $O(|H^*|)$. Since there are $f+2 = O(\log{r})$ subsets, we conclude that $\sum_{e \in E} cost(e) \leq O(\log{r})|H^*|.$

Let $E_0 = \{e: cost(e) = 0\}$ and $E_{f+1} = H_2$. Note that all the edges not in $E_0$ and $E_{f+1}$ are edges that were 2-spanned in the algorithm by the candidate star they vote for. We divide these edges to $f$ subsets as follows. Let $E_1 = \{e \not \in H_2: 0 < cost(e) \leq \frac{2}{r} \}$, and for $2 \leq j \leq f$, let $E_j = \{e \not \in H_2: \frac{2^{j-1}}{r} < cost(e) \leq \frac{2^j}{r} \}.$
For each edge $e$, it holds that $cost(e) \leq 1$, since the density of stars added to $H_1$ during the algorithm is at least 1, and since we defined $cost(e) = 1$ for edges $e \in H_2$.
Hence, for each edge $e$, we have $cost(e) \leq \frac{2^f}{r}$, which gives $E = \cup_{j=0}^{f+1} E_f$.

\begin{lemma} \label{Hci}
For every $0 \leq j \leq f+1$, $\sum_{e \in E_j} cost(e) = O(|H^*|).$
\end{lemma}

\begin{proof}
For $j=0$, the claim holds trivially.
For $j=1$, it holds that
$\sum_{e \in E_1} cost(e) \leq \frac{2}{r} \cdot |E| \leq 2n = O(|H^*|),$ where the last equality follows from the fact that any $2$-spanner for $G$ has at least $n-1$ edges, since $G$ is connected.

For $2 \leq j \leq f$, let $H^*_j$ be the set of edges of a minimum 2-spanner for $E_j$. For each vertex $v$, let $S^*_j(v)$ be the full $v$-star in $H^*_j$. We define $Stars_j = \{S^*_j(v)\}_{v \in V}.$
We next show that $\sum_{e \in E_j} cost(e) \leq 9 \sum_{S \in Stars_j} |S|$. To prove this, we write $\sum_{e \in E_j} cost(e) = \sum_{e \in E_j \cap Stars_j} cost(e) + \sum_{e \in E_j \setminus Stars_j}cost(e).$ Since $cost(e) \leq 1$, we get $\sum_{e \in E_j \cap Stars_j} cost(e) \leq \sum_{S \in Stars_j} |S|$. 

We now show that $\sum_{e \in E_j \setminus Stars_j}cost(e) \leq 8\sum_{S \in Stars_j} |S|$.
Consider a specific star $S \in Stars_j$, and let $(e_1,...,e_{\ell})$ be the edges of $E_j$ 2-spanned by $S$ according to the order in which they were 2-spanned in the algorithm, breaking ties arbitrarily. Note that all the edges in $E_j$ for $1 \leq j \leq f$ are 2-spanned in the algorithm as explained above.
The density of $S$ at the beginning of the iteration in which $e_1$ is 2-spanned is at least $\frac{\ell}{|S|}$, since $S$ may 2-span additional edges that are not in $E_j$. Since all the candidates that 2-span an edge have the same rounded density because they all have maximal rounded density in their 2-neighborhood, it holds that the density of the star $S_v$ that 2-spans $e_1$ is at least $\frac{\ell}{4|S|}$, as $v$ chooses a star with density at least $\frac{\rho_v}{4} \geq \frac{\ell}{4|S|}$. Hence, $cost(e_1) \leq \frac{4|S|}{\ell}$. This gives, $\ell \leq \frac{4|S|}{cost(e_1)} \leq \frac{r}{2^{j-1}} 4|S|,$ where the last inequality follows since $e_1 \in E_j$.
Note that for each edge $e \in E_j$, $cost(e) \leq \frac{2^j}{r}$. Therefore, $$\sum_{i=1}^{\ell} cost(e_i) \leq \frac{2^j}{r} \ell \leq \frac{2^j}{r} \cdot \frac{r}{2^{j-1}} 4|S| = 8|S|.$$

Let $C_S$ be the set of edges of $E_j$ 2-spanned by the star $S$. Since $H^*_j$ is a 2-spanner for $E_j$, every edge $e \in E_j \setminus Stars_j$ is 2-spanned by at least one star $S \in Stars_j$. Summing over all the stars in $Stars_j$ gives
$$\sum_{e \in E_j \setminus Stars_j} cost(e) \leq \sum_{S \in Stars_j} \sum_{e \in C_S} cost(e) \leq 8\sum_{S \in Stars_j} |S|.$$

Note that $\sum_{S \in Stars_j} |S| = 2|H^*_j|$, since each edge of $H^*_j$ is included in exactly two stars in $Stars_j$. In addition $|H^*_j| \leq |H^*|$ since $H^*$ covers all the edges of $E$, and in particular all the edges of $E_j$, and $H^*_j$ is the minimum 2-spanner of $E_j$.
This gives $\sum_{e \in E_j} cost(e) = O(|H^*|)$, which completes the proof for $2 \leq j \leq f$.

For $j=f+1$, we define $H^*_j$ and $Stars_j$ as before. Let $S \in Stars_j$, and let $(e_1,...,e_{\ell})$ be the edges of $E_j=H_2$ that are 2-spanned by $S$ according to the order in which they were added to $H_2$ in the algorithm, breaking ties arbitrarily.
It must hold that $\ell \leq |S|$, as otherwise the density of $S$ is greater than one at the iteration in which $e_1$ is added to $H_2$, which contradicts the algorithm. This gives $\sum_{i=1}^{\ell} cost(e_i) \leq \ell \leq |S|.$  Following the same arguments for the case $2 \leq j \leq f$, we get $\sum_{e \in E_j} cost(e) = O(|H^*|)$, which completes the proof. 
\end{proof}

Lemmas \ref{cost} and \ref{Hci} give
$|H| \leq 8 \sum_{e \in E} cost(e) \leq O(\log{r}) |H^*|$, which proves the following claimed approximation ratio.

\begin{lemma} \label{approx}
The approximation ratio of the algorithm is $O(\log{\frac{m}{n}})$.
\end{lemma}

\subsubsection{Time complexity} \label{sec:time}

We now show that our algorithm completes in $ O(\log{n}\log{\Delta})$ rounds, w.h.p.
In \cite{jia2002efficient, rajagopalan1998primal}, a potential function argument is given for analyzing the set cover and minimum dominating set problems that are addressed. We analyze our algorithm along a similar argument, but our algorithm necessitates a more intricate analysis, mainly due to the fact that each vertex may be the center of multiple stars that are added during the algorithm, rather than being chosen only once for a dominating set. The latter may contain at most $n$ vertices, while for the spanner constructed by our algorithm there are initially $n2^{\Delta}$ possible stars which may constitute it. Nevertheless, we show how to get a time complexity of $O(\log{n}\log{\Delta})$ rounds for our minimum 2-spanner algorithm, which matches the time complexity of the above set cover and dominating set algorithms.

The crucial component in proving our small time complexity is showing that as long as the rounded density of $v$ does not change between iterations, $v$ always chooses a star $S_v$ that is equal to or contained in the star that it chooses in the previous iteration.
As explained in Section \ref{choose}, if the rounded density of $v$ is the same in iterations $i$ and $i+1$, $v$ tries to choose a star $S_v^{i+1}$ which is contained in $S_v^i$. We show that this is always the case.

The following will be useful in our analysis.

\begin{observation} \label{obs2}
Let $x_1,x_2,...,x_n$ be non-negative numbers, and let $y_1,y_2,...,y_n$ be positive numbers, then $$\min_i{\left \{ \frac{x_i}{y_i} \right \} } \leq \frac{\sum_{i=1}^n x_i}{\sum_{i=1}^n y_i} \leq \max_i{\left \{ \frac{x_i}{y_i} \right \} }.$$
In addition, the inequalities become equalities only if for all $j$, $\frac{x_j}{y_j} = \min_i{\left \{ \frac{x_i}{y_i} \right \} } =\max_i{\left \{ \frac{x_i}{y_i} \right \} }$.
\end{observation}

Observation \ref{obs2} follows from writing $\sum_{i=1}^n x_i = \sum_{i=1}^n \frac{x_i}{y_i} \cdot y_i.$
We now prove the following.

\begin{claim} \label{iteration_p}
Let $v$ be a candidate star in iteration $i$. If $\tilde{\rho}(v,H_v^{i})=\tilde{\rho}(v,H_v^{i+1}) = \rho$,
then $v$ chooses a star contained in $S_v^i$ in iteration $i+1$.
\end{claim}

\begin{proof}
Assume to the contrary that there is an iteration $i$ such that $\tilde{\rho}(v,H_v^{i})=\tilde{\rho}(v,H_v^{i+1}) = \rho$, and there is no star contained in $S_v^i$ with density at least $\frac{\rho}{4}$ with respect to $H_v^{i+1}$.
Let $i_0$ be the first iteration in which $\tilde{\rho}(v,H_v^{i_0})=\rho$, and let $i'$ be the first iteration after $i_0$ where $\tilde{\rho}(v,H_v^{i'})=\rho$ and there is no star contained in $S_v^{i'-1}$ with density at least $\frac{\rho}{4}$ with respect to $H_v^{i'}$. Let $S^*$ be the densest $v$-star with respect to $H_v^{i'}$. Then $\rho(S^*, H_v^{i'}) \geq \frac{\rho}{2}$ since $\tilde{\rho}(v,H_v^{i'})=\rho$.
Let $S_0$ be the full $v$-star, and let $(S_1,S_2,...,S_k)$ be the sequence of stars chosen by $v$ between iteration $i_0$ and iteration $i'-1$ in the order in which they were chosen. For all $0 \leq j \leq k$, it holds that $S_j \subseteq S_{j-1}$, since $i'$ is the first iteration in which this does not hold. 

We next show by induction that $S^* \subseteq S_j$ for all $0 \leq j \leq k$. In particular this will give $S^* \subseteq S_k$. Hence, at iteration $i'$ there is a star contained in $S_v^{i'-1} = S_k$ with density at least $\frac{\rho}{4}$, in contradiction to the definition of $i'$.

For $j=0$, the claim holds trivially since $S_0$ is the full $v$-star. 

Assume that $S^* \subseteq S_{j-1}$, and assume to the contrary that $S^* \nsubseteq S_j$. Note that both $S_j$ and $S^*$ are contained in $S_{j-1}$ by the induction hypothesis.
Let $j'$ be the iteration in which $S_j$ is chosen. Since $S^* \nsubseteq S_j$, we can write $S^* = S_1 \cup S_2$ where $S_1 = S^* \cap S_j$ and $S_2 = S^* \setminus S_j$. It holds that $\rho(S^*,H_v^{j'}) \geq \rho(S^*,H_v^{i'}) \geq \frac{\rho}{2}$. We can write $\rho(S^*,H_v^{i'}) = \frac{|C_1|+|C_2|+|C_{12}|}{|S_1|+|S_2|}$ where $C_1$ are edges of $H_v^{i'}$ 2-spanned by $S_1$, $C_2$ are edges 2-spanned by $S_2$, and $C_{12}$ are edges 2-spanned by $S^*$ with one endpoint in $S_1$ and one endpoint in $S_2$. Since $S^*$ is the densest star with respect to $H_v^{i'}$ it follows that $\frac{|C_2|+|C_{12}|}{|S_2|} \geq \rho(S^*,H_v^{i'}) \geq \frac{\rho}{2}$, as otherwise by Observation \ref{obs2}, $\frac{|C_1|}{|S_1|}>\rho(S^*,H_v^{i'})$, which shows that $S_1$ is a denser star than $S^*$.
This shows that at least one of $\frac{|C_2|}{|S_2|}$ and $\frac{|C_{12}|}{|S_2|}$ is at least $\frac{\rho}{4}$. 

In the first case, $\rho(S_2,H_v^{j'}) \geq \rho(S_2,H_v^{i'}) = \frac{|C_2|}{|S_2|} \geq \frac{\rho}{4}$, which shows that $S_2$ is a disjoint star to $S_j$ with density at least $\frac{\rho}{4}$ that is contained in $S_{j-1}$. In the second case, $\frac{|C_{12}|}{|S_2|} \geq \frac{\rho}{4}$. For an edge $e=\{v,u\} \in S_2$, denote by $C^e_{12}$ all the edges of $C_{12}$ with endpoint $u$. It follows that there is an edge $e \in S_2$ such that $|C^e_{12}| \geq \frac{\rho}{4}$. By Observation \ref{obs2} and since the density of $S_j$ is at least $\frac{\rho}{4}$ we get $\rho(S_j \cup \{e\}) \geq \frac{|C_j|+|C^e_{12}|}{|S_j|+1} \geq \frac{\rho}{4}$, where $C_j$ are the edges 2-spanned by $S_j$. Either way we get a contradiction to the definition of $S_j$.
This completes the proof.
\end{proof}

The rest of the analysis is based on a potential function argument which is described in \cite{jia2002efficient, rajagopalan1998primal} for the set cover and minimum dominating set problems. Let $\rho=\max_{v \in V}{\rho_v}$ at the beginning of iteration $i$. We define the potential function $\phi = \sum_{v:\rho_v=\rho}|C_v|$, where $C_v$ is the set of edges of $H_v$ that are 2-spanned by the star $S_v$ which $v$ chooses at iteration $i$.
Note that the potential function may increase between iterations if the value of $\rho$ changes. However, since we round the values $\rho_v$ to powers of two, there may be only $O(\log{\Delta})$ different values for $\rho$. The obstacle is that $\phi$ may increase between iterations even if the value of $\rho$ does not change, because a vertex $v$ might change the stars $S_v$ in different iterations.
However, by Claim \ref{iteration_p}, as long as the rounded density of the vertex remains the same among iterations, it always chooses a star that is contained in the star that it chooses in the previous iteration. Hence, the size of the set of edges $C_v$ can only decrease between the end of the last iteration to the beginning of the next one. It follows that as long as $\rho$ does not change, the value of $\phi$ can only decrease between iterations.
Our goal is to show that if the value of $\rho$ does not change between iterations, the potential function $\phi$ decreases by a multiplicative factor between iterations in expectation. Having this, we get a time complexity of $O(\log{n}\log{\Delta})$ rounds w.h.p.

The following lemma shows that if the value of $\rho$ does not change between iterations, the potential function $\phi$ decreases by a multiplicative factor between iterations in expectation. The proof follows the lines of the proofs in \cite{jia2002efficient, rajagopalan1998primal}, and is included here for completeness.

We say that an iteration is \emph{legal} if the random numbers $r_v$ chosen by the candidates in this iteration are different.

\begin{lemma} \label{dec}
If $\phi$ and $\phi'$ are the potentials at the beginning and end of a legal iteration, then $E[\phi'] \leq c \cdot \phi$ for some positive constant $c<1$.
\end{lemma}

In order to prove Lemma \ref{dec} we need the following definitions. Let $s(e)$ be the number of candidates that 2-span the edge $e$. For a candidate $v$, we sort the edges in $C_v$ according to $s(e)$ in non-increasing order. Let $T(v)$ and $B(v)$ be the sets of the first $\lceil |C_v|/2 \rceil$ edges, and the last $\lceil |C_v|/2 \rceil$ edges in the sorted order, respectively. Indeed, if $|C_v|$ is odd, the sets $T(v)$ and $B(v)$ share an edge.

For a pair $(S_v,e)$, where $S_v$ is a candidate star that 2-spans $e$, we say that $(S_v,e)$ is \emph{good} if $e \in T(v)$.
We next show that if $e \in T(v)$ chooses $S_v$ in a legal iteration, then the star $S_v$ is added to the spanner with constant probability.

\begin{claim} \label{greater}
Let $i$ be a legal iteration. If $e,e'$ are both 2-spanned by $S$ in iteration $i$, and $s(e) \geq s(e')$, then $Pr[e'\ chooses\ S| e\ chooses\ S] \geq \frac{1}{2}$.
\end{claim}

\begin{proof}
Let $N_e,N_{e'},N_b$ be the number of candidates that 2-span $e$ but not $e'$, $e'$ but not $e$, and both $e$ and $e'$, respectively.
$$Pr[e'\ chooses\ S| e\ chooses\ S] = \frac{Pr[e\ and\ e'\ choose\ S]}{Pr[e\ chooses\ S]} = \frac{\frac{1}{N_e + N_{e'} + N_b}}{\frac{1}{N_e + N_b}} = \frac{|N_e|+|N_b|}{|N_e|+|N_{e'}|+|N_b|}.$$
It holds that $N_e \geq N_{e'}$ since $s(e) \geq s(e')$. This gives,
$$Pr[e'\ chooses\ S| e\ chooses\ S] = \frac{|N_e|+|N_b|}{|N_e|+|N_{e'}|+|N_b|} \geq \frac{|N_e|+|N_b|}{2|N_e|+|N_b|} \geq \frac{1}{2}.$$
\end{proof}

\begin{claim}
If $(S_v,e)$ is a good pair in a legal iteration $i$, then $Pr[S_v\ is\ chosen|e\ chooses\ S_v] \geq \frac{1}{3}$.
\end{claim}

\begin{proof}
Assume that $e$ chooses $S_v$. Denote by $X$ the number of edges in $B(v)$ that choose $S_v$, and let $X'=|B(v)|-X$. Note that $e \in T(v)$ since $(S_v,e)$ is good, therefore $s(e) \geq s(e')$ for any edge $e' \in B(v)$.
By Claim \ref{greater}, any edge $e' \in B(v)$ chooses $S_v$ with probability at least $\frac{1}{2}$. Hence, $E[X] \geq \frac{|B_v|}{2}$. Equivalently, $E[X'] \leq \frac{|B_v|}{2}$. Using Markov's inequality we get
$$Pr[X < \frac{|B_v|}{4}]=Pr[X' > \frac{3}{4}|B_v|] \leq Pr[X' \geq \frac{3}{2}E[X']] \leq \frac{2}{3}.$$
Hence, we get $Pr[X \geq \frac{|B_v|}{4}] \geq \frac{1}{3}$. Since $|B_v| \geq \frac{|C_v|}{2}$, it holds that $X \geq \frac{|C_v|}{8}$ with probability at least $\frac{1}{3}$. In this case, at least $\frac{|C_v|}{8}$ edges choose $S_v$, and it is added to the spanner. This completes the proof.
\end{proof}

We can now bound the value of the potential function, by proving Lemma \ref{dec}.

\begin{proof} [Proof of Lemma \ref{dec}.]
Let $\phi$ and $\phi'$ be the values of the potential function at the beginning and end of a legal iteration $i$.
It holds that $\phi = \sum_{v:\rho_v=\rho}|C_v|= \sum_{(S_v,e)} 1 = \sum_{e} s(e)$, where we sum over all the edges 2-spanned by candidates having rounded density $\rho$, and over all the pairs $(S_v,e)$ where $v$ is a candidate having rounded density $\rho$ that 2-spans $e$. Note that the rounded density of all the candidates that 2-span an edge $e$ is the same, since they have maximal rounded density in their 2-neighborhood. If the edge $e$ chooses the star $S_v$, and the star $S_v$ is added to the spanner, $\phi$ decreases by $s(e)$. We ascribe this decrease to the pair $(S_v,e)$. Since $e$ chooses only one candidate, any decrease in $\phi$ is ascribed only to one pair. Hence, we get
\begin{flalign*}
E[\phi - \phi'] &\geq \sum_{(S_v,e)} Pr[e\ chooses\ S_v,S_v\ is\ chosen]\cdot s(e) \\
&\geq \sum_{(S_v,e)\ is\ good} Pr[e\ chooses\ S_v]\cdot Pr[S_v\ is\ chosen|e\ chooses\ S_v]\cdot s(e) \\
&\geq \sum_{(S_v,e)\ is\ good} \frac{1}{s(e)} \cdot \frac{1}{3} \cdot s(e) = \frac{1}{3} \sum_{(S_v,e)\ is\ good} 1.
\end{flalign*}
Since at least half of the pairs are good, we get $E[\phi - \phi'] \geq \frac{1}{6} \phi$, or equivalently $E[\phi'] \leq \frac{5}{6} \phi$, which completes the proof.
\end{proof}

In conclusion, we get the following.

\begin{lemma} \label{time}
The time complexity of the algorithm is $O(\log{n} \log{\Delta})$ rounds w.h.p.
\end{lemma}

\begin{proof}
It holds that the maximum density of a star of size $k$ is at most $O(k^2)$ and the algorithm terminates when the maximum density is at most 1. Since densities are rounded to powers of 2, 
$\rho=\max_{v \in V}{\rho_v}$ may obtain at most $O(\log{\Delta})$ values, where $\Delta$ is the maximum degree. In addition, by Lemma \ref{dec}, if $\rho$ has the same value at iterations $j$ and $j+1$, and $j$ is a legal iteration, then the value of $\phi$ decreases between these iterations by a factor of at least $1/c$ in expectation. Since the random numbers $r_v$ are chosen from $\{1,...,n^4\}$, they are different w.h.p, giving that if $\rho$ has the same value in any two consecutive iterations then the value of $\phi$ decreases between these iterations by a constant factor in expectation.
Since $\phi \leq n^3$, after $O(\log{(n^3)})=O(\log{n})$ iterations in expectation, the value of $\rho$ must decrease. This shows that the time complexity is $O(\log{n} \log{\Delta})$ rounds in expectation. A Chernoff bound then gives that this also holds w.h.p.
\end{proof}

Lemma \ref{approx} and Lemma \ref{time} complete the proof of Theorem \ref{minimumSpanner}.

\subsection{Additional 2-spanner approximations} \label{sec:additional}

Here we show that the algorithm extends easily to the following variants: the directed 2-spanner problem, the weighted 2-spanner problem and the client-server 2-spanner problem. We describe the differences in the algorithm and analysis in each of these cases.

\subsubsection{Directed 2-spanner approximation}

In the directed case we consider $\textit{directed}$ stars. A $v$-star \emph{2-spans} a directed edge $(u,w)$ if it includes the directed edges $(u,v),(v,w)$. A $v$-star may include both ingoing and outgoing edges of $v$. The definition of densities follows the definition in the undirected case. 

In order to give an algorithm that requires only polynomial local computations for the directed variant, we show how to approximate the rounded density and the densest star in the directed case. The rest of the analysis follows the undirected case, and gives the following.

\begin{theorem}
There is a distributed algorithm for the directed 2-spanner problem in the \local model that guarantees an approximation ratio of $O(\log{\frac{m}{n}})$, and takes $O(\log{n} \log{\Delta})$ rounds w.h.p.
\end{theorem}

To compute an approximation for the densest (directed) $v$-star, we look at all the edges between neighbors of $v$, and remove all the directed edges $(u,w)$ that cannot be 2-spanned by a $v$-star (such an edge can be 2-spanned by a $v$-star, only if the two directed edges $(u,v),(v,w)$ exist in the graph). Now we ignore the directions of edges and compute the densest $v$-star as in the undirected case. Let $S_v$ be the star computed.
Let $\rho_U=\rho_U(S_v)$ be the undirected density of $v$ (when ignoring edges that cannot be 2-spanned by a directed path), and let $\rho_D$ be the directed density of $v$. We will show that $\frac{\rho_U}{2} \leq \rho_D \leq \rho_U$, and that $\frac{\rho_U}{2} \leq \rho_D(S_v)$, which shows that $S_v$ gives a 2-approximation for the densest directed $v$-star. When computing $\rho_D(S_v)$ we replace each undirected edge $\{v,u\}$ in $S_v$ by the two directed edges $(v,u),(u,v)$ if both of them exist in the graph, or by the one that exists otherwise.

\begin{claim} \label{dir_density}
$\frac{\rho_U}{2} \leq \rho_D(S_v)$.
\end{claim}

\begin{proof}
Let $C_v$ be the edges 2-spanned by $S_v$ in the undirected case. Since $S_v$ is the densest undirected $v$-star, $\rho_U = \rho_U(S_v) = \frac{|C_v|}{|S_v|}$. The directed star $S_v$ 2-spans all the edges in $C_v$ (some of them may be counted twice in the directed case, which only increases the density), and it contains at most twice edges because we replaced each undirected edge by at most two edges, which gives $\rho_D(S_v) \geq \frac{|C_v|}{2|S_v|}=\frac{\rho_U}{2}$. 
\end{proof}

\begin{claim} \label{dir_density2}
$\frac{\rho_U}{2} \leq \rho_D \leq \rho_U$.
\end{claim}

\begin{proof}
By Claim \ref{dir_density}, $\frac{\rho_U}{2} \leq \rho_D(S_v) \leq \rho_D$. We next show $\rho_D \leq \rho_U$. Let $S_D$ be the densest directed $v$-star, and let $C_D$ be the directed edges 2-spanned by $S_D$. We write $C_D = C_1 \cup C_2$ where an edge $(u,w) \in C_D$ is in $C_2$ if and only if the edge $(w,u)$ is also in $C_D$. We write $S_D = S_1 \cup S_2$ where an edge $(v,u) \in S_D$ is in $S_2$ if and only if the edge $(u,v)$ is also in $S_D$. Now we look at the undirected density of $S_D$ (if we have two directed edges $(w,u),(u,w)$ they are replaced by one undirected edge). $$\rho_U(S_D) = \frac{|C_1|+\frac{|C_2|}{2}}{|S_1| + \frac{|S_2|}{2}} = \frac{\frac{|C_1|}{2}+ \frac{|C_1|+|C_2|}{2}}{\frac{|S_1|}{2}+\frac{|S_1|+|S_2|}{2}} \geq \min{\left\lbrace  \frac{|C_1|}{|S_1|},\frac{|C_1|+|C_2|}{|S_1|+|S_2|}\right\rbrace }=\min{\left\lbrace \frac{|C_1|}{|S_1|},\rho_D\right\rbrace },$$
where the second inequality follows from Observation \ref{obs2}, and the last equality follows since $S_D$ is the densest directed $v$-star, and its directed density equals $\frac{|C_1|+|C_2|}{|S_1|+|S_2|}$. 

We next show that $\frac{|C_1|}{|S_1|} \geq \rho_D$. Since $\rho_D = \frac{|C_1|+|C_2|}{|S_1|+|S_2|}$, if $\frac{|C_1|}{|S_1|} < \rho_D$, we get by Observation \ref{obs2} that $\frac{|C_2|}{|S_2|}> \rho_D$. Note that the directed star $S_2$ 2-spans all the directed edges in $C_2$ (and it may 2-span additional edges), because all the edges in $C_2$ appear in both directions, which means that the paths that 2-span them also appear in both directions in $S_D$, which means that all the edges of these paths are in $S_2$. This shows that $S_2$ is a directed star with density greater than $\rho_D$, in contradiction to the definition of $\rho_D$. In conclusion $\frac{|C_1|}{|S_1|} \geq \rho_D$, which shows that $\rho_U \geq \rho_U(S_D) \geq \min{\left\lbrace \frac{|C_1|}{|S_1|},\rho_D\right\rbrace } = \rho_D.$
\end{proof}

Claims \ref{dir_density} and \ref{dir_density2} show that we can approximate the directed density and the densest directed star using polynomial local computations. Having this, we can adapt the algorithm to the directed 2-spanner problem. We approximate the directed density of $v$ with $\rho_D(S_v)$, which gives a 2-approximation. Then, we round the value of it to the closest power of two that is greater than $\rho_D(S_v)$, and denote the rounded value by $\rho_v$, this is a 2-approximation to the rounded density of $v$.\footnote{Since we compute an approximation to the density, the value of $\rho_v$ may increase between iterations. To avoid such cases, we always define it to be the minimum between the value in the last iteration and the value computed in the current iteration. This is always a 2-approximation for the rounded density since the density can only decrease between iterations.} While the value of $\rho_v$ remains the same we choose stars similarly to the undirected case, the only difference is that when we look for a dense disjoint star, we do not necessarily find the densest directed disjoint star but a 2-approximation for it. For the analysis to work, we need to look for disjoint stars with density at least $\frac{\rho_v}{8}$ and not $\frac{\rho_v}{4}$ as in the undirected case, and we add edges or disjoint stars to the star $S$ we choose as long as the density of $S$ is at least $\frac{\rho_v}{8}$. The rest of the analysis is similar, the constants change sightly since we choose stars that are less dense, and work with an approximation to the density. 

\subsubsection{Weighted 2-spanner approximation}

In the weighted case, the cost of a spanner is $w(H)$, rather than $|H|$ as in the unweighted case.
This requires several changes in the algorithm and the analysis.
Let $W$ be the ratio between the maximum and minimum positive weights of an edge. We show the following.

\begin{theorem} \label{weightedAlg}
There is a distributed algorithm for the weighted 2-spanner problem in the \local model that guarantees an approximation ratio of $O(\log{\Delta})$, and takes $O(\log{n} \log{(\Delta W)})$ rounds w.h.p.
\end{theorem}

We next describe the differences in the weighted case.
For a star $S$, we define $w(S) = \sum_{e \in S} w(e)$. If $w(S) \neq 0$, we define $\rho(S,H)=\frac{|C_S|}{w(S)},$ where $C_S$ is the set of edges of $H$ 2-spanned by the star $S$. If $w(S) = 0$, we define $\rho(S,H) = 0$. We emphasize that we take the number of potentially 2-spanned edges and not the sum of their weights, since, intuitively, all edges need to be covered (as opposed to taking the sum of weights of the edges of the star, which is due to the need to optimize the cost of the spanner).

In the beginning of the algorithm we add all the edges of weight $0$ to the spanner. By doing this, all the edges covered by stars $S$ of weight $0$ are already covered. Hence, the algorithm should only consider stars for which $w(S) > 0$. When we round the densities to the closest power of two, we include also negative powers of two, since the density of a star may be smaller than $1$, depending on the weights. A slight difference is that now a vertex terminates if the density in its 2-neighborhood is at most $\frac{1}{w_{max}}$ where $w_{max}$ is the maximal weight of an edge adjacent to a vertex in its 2-neighborhood. In such case, it adds to the spanner all the edges adjacent to it that are still not covered. We denote by $H_2$ all these edges.
The rest of the algorithm is the same, according to the new definition of $\rho$. As was observed in the sequential algorithm for the weighted case \cite{kortsarz2001hardness}, we can find the densest star in the weighted case using flow techniques as well.

We next describe the differences in the analysis. The cost of the solution obtained by the algorithm is now $w(H)$, and the cost of an optimal solution is $w(H^*)$. We give to edges $e \not \in H_2$ a cost as in the unweighted case, but depending on the new definition of the density $\rho$. In addition, for edges $e \in H_2$ we define $cost(e) = w(e)$. Our goal is to show that $$w(H) \leq 8 \sum_{e \in E} cost(e) \leq O(\log{\Delta}) w(H^*).$$

The proof that $w(H) \leq 8 \sum_{e \in E} cost(e)$ is the same as the proof of Lemma \ref{cost} with minor changes. Note that $w(H_2) = \sum_{e \in H_2} cost(e)$ by definition. Now $w(H_1) \leq \sum_{S \in Stars} w(S)$, for the same reason as in the unweighted case. In addition, the new definition of $\rho$, gives that $w(S) \leq 8\sum_{e \in Votes(S)} cost(e),$ and the rest of the proof follows.

However, the difference from the unweighted case is that we can no longer show an approximation ratio of $O(\log{\frac{m}{n}})$. This is because the density of stars added in the algorithm may now be smaller than 1, and because the weight of an optimal 2-spanner may be smaller than $n-1$.
Still, we show an approximation ratio of $O(\log{\Delta})$. Some elements of our analysis have similar analogues in the classic analysis of the greedy set cover algorithm \cite{johnson1974approximation, chvatal1979greedy, lovasz1975ratio}. First, instead of Lemma \ref{Hci} we show the following.

\begin{lemma} \label{weighted}
$\sum_{e \in E} cost(e) \leq O(\log{\Delta})w(H^*)$.
\end{lemma}

\begin{proof}
First, we show that $\sum_{e \in E'} cost(e) \leq O(\log{\Delta})w(H^*),$ where $E' = E \setminus (H_2 \cap E_0)$ and $E_0$ are edges with cost $0$. Edges in $E_0$ clearly do not affect $\sum_{e \in E} cost(e)$. All the edges not in $E_0$ or $H_2$ are 2-spanned in the algorithm, as in the unweighted case.

For each vertex $v$, let $S^*(v)$ be the full $v$-star in $H^*$. We define $Stars^* = \{S^*(v)\}_{v \in V}.$
Consider a star $S \in Stars^*$ and let $(e_1,...,e_{\ell})$ be the sequence of edges 2-spanned by $S$ according to the order in which they are 2-spanned in the algorithm. Assume first that $w(S) \neq 0$.
The density of $S$ at the beginning of the iteration in which $e_1$ is 2-spanned is $\frac{\ell}{w(S)}$. All the candidates that 2-span $e_1$ have the same rounded density since they all have maximal rounded density in their 2-neighborhood. In particular, the density of the star that 2-spans $e_1$ is at least $\frac{\ell}{4w(S)}$, as $v$ chooses a star with density at least $\frac{\rho_v}{4} \geq \frac{\ell}{4w(S)}$. Hence, $cost(e_1) \leq \frac{4w(S)}{\ell}$.
Similarly, the density of $S$ at the beginning of the iteration in which $e_j$ is 2-spanned is at least $\frac{\ell - j +1}{w(S)}$, which gives $cost(e_j) \leq \frac{4w(S)}{\ell - j + 1}$.
This gives, $$\sum_{j=1}^{\ell} cost(e_j) \leq 4w(S) \cdot \sum_{j=1}^{\ell} \frac{1}{\ell - j + 1} = O(\log{\ell})w(S) = O(\log{\Delta})w(S).$$
The last equality is because the number of edges $\ell$ 2-spanned by a star is at most $\Delta^2$.

For a star $S \in Stars^*$ such that $w(S)=0$, note that $cost(e)=0$ for all the edges 2-spanned by $S$, since they are all covered at the beginning of the algorithm without voting for any candidate. Hence, we get in this case $\sum_{j=1}^{\ell} cost(e_j) = 0 = O(\log{\Delta})w(S).$

We now write $\sum_{e \in E'} cost(e) = \sum_{e \in E' \cap Stars^*} cost(e) + \sum_{e \in E' \setminus Stars^*} cost(e)$. It holds that $\sum_{e \in E' \cap Stars^*} cost(e) \leq \sum_{S \in Stars^*} w(S).$ We next bound  $\sum_{e \in E' \setminus Stars^*} cost(e).$ 

Let $C_S$ be the set of edges 2-spanned by the star $S$. Since $H^*$ is a 2-spanner, every edge $e \in E' \setminus Stars^*$ is 2-spanned by at least one star $S \in Stars^*$. Summing over all the stars in $Stars^*$ we get,
$$\sum_{e \in E' \setminus Stars^*} cost(e) \leq \sum_{S \in Stars^*} \sum_{e \in C_S} cost(e) \leq O(\log{\Delta})\sum_{S \in Stars^*} w(S).$$
In conclusion, $\sum_{e \in E'} cost(e) = O(\log{\Delta})\sum_{S \in Stars^*} w(S).$

It holds that $\sum_{S \in Stars^*} w(S) = 2w(H^*)$ since each edge of $H^*$ is included in exactly two stars. This gives, $\sum_{e \in E'} cost(e) = O(\log{\Delta})w(H^*)$.

To complete the proof, we bound $\sum_{e \in H_2} cost(e).$ Let $H_2^*$ be an optimal spanner for $H_2$. We define $Stars^*$ as before, with respect to $H_2^*$. Let $S \in Stars^*$ and let $(e_1,...,e_{\ell})$ be the sequence of edges of $H_2$ 2-spanned by $S$ according to the order in which they are added to $H_2$ in the algorithm. From the definition of $H_2$ it must hold that $\frac{\ell}{w(S)} \leq \frac{1}{w_{max}}$ where $w_{max}$ is the maximal weight in the 2-neighborhood of $e_1$ (which in particular contains the star $S$ and all the edges 2-spanned by it), as otherwise $e_1$ was not added to $H_2$. This gives $\sum_{i=1}^{\ell} cost(e) = \sum_{i=1}^{\ell} w(e) \leq \ell \cdot w_{max} \leq w(S)$. Following the same arguments as before, this gives $\sum_{e \in H_2} cost(e) \leq O(w(H_2^*))$. Since $|H_2^*| \leq |H^*|$, we get $\sum_{e \in E} cost(e) = \sum_{e \in E'} cost(e) + \sum_{e \in H_2} cost(e) = O(\log{\Delta})w(H^*)$. This completes the proof.
\end{proof}

In conclusion, we get $w(H) \leq 8 \sum_{e \in E} cost(e) \leq O(\log{\Delta}) w(H^*)$, which completes the proof of the $O(\log{\Delta})$-approximation ratio, giving the following lemma.

\begin{lemma}
The approximation ratio of the algorithm is $O(\log{\Delta})$.
\end{lemma}

To prove the round complexity,
there are minor changes in the proof of Claim \ref{iteration_p}. First, we replace the size of a star $|S|$ by its cost $w(S)$ in order to work with the new definition of $\rho$. Note that adding edges of weight $0$ to a star can only increase its density, which shows that all the $v$-stars chosen in the algorithm, and in particular the star $S_j$, contain all the edges of weight $0$ adjacent to $v$. This shows that the star $S_2 = S^* \setminus S_j$ includes only edges with positive weight. The proof carries over if $w(S_1) \neq 0$. If $w(S_1) = 0$, then all the edges 2-spanned by $S_1$ are already 2-spanned at the  beginning of the algorithm which shows $\frac{\rho}{2} \leq \rho(S^*,H_v^{i'}) = \frac{|C_1|+|C_2|+|C_{12}|}{w(S_1)+w(S_2)} = \frac{|C_2|+|C_{12}|}{w(S_2)}$, as needed. In addition, in the second case of the proof instead of showing that there is an edge $e \in S_2$ with $|C^e_{12}| \geq \frac{\rho}{4}$, we show that $\frac{|C^e_{12}|}{w(e)} \geq \frac{\rho}{4}$.

The number of possible densities depends on the weights, in the following way. Let $W_{max},W_{min}$ be the maximum and the minimum positive weights of an edge. Recall that $W = \frac{W_{max}}{W_{min}}$.
The maximum density of a star is at most $\frac{\Delta^2}{W_{min}}$ since a star 2-spans at most $\Delta^2$ edges. In addition, the algorithm terminates when the maximum density is $\frac{1}{W_{max}}$. Since we round the densities to powers of two, there may at most $O(\log{\Delta W})$ different non-zero values for the densities. The rest of the proof is exactly the same as in the unweighted case.

\subsubsection{Client-server 2-spanner approximation}

Recall that in the Client-Server 2-spanner problem, the edges of the graph are divided to two types: clients and servers, and the goal is to cover all the client edges with server edges.

Let $C$ be the set of client edges, let $V(C)$ be all the vertices that touch client edges, and let $\Delta_S$ be the maximum degree in the subgraph of $G$ that includes all the server edges. We show the following.

\begin{theorem} \label{csThm}
There is a distributed algorithm for the client-server 2-spanner problem in the \local model that guarantees an approximation ratio of $O(\min\{\log{\frac{|C|}{|V(C)|}},\log{\Delta_S}\})$, and takes $O(\log{n} \log{\Delta_S})$ rounds w.h.p.
\end{theorem}

There are slight differences in the algorithm. First, throughout the algorithm and analysis, we consider only stars composed of server edges, and for each such star we define $\rho(S,H)=\frac{|C_S|}{|S|}$, where $C_S$ is the set of \textit{client} edges of $H$ 2-spanned by the star $S$. The set of edges $H_v$ that a vertex $v$ maintains consists only of client edges 2-spanned by the star that includes all the server edges adjacent to $v$. Now $v$ terminates if the maximal density in its 2-neighborhood is below $\frac{1}{2}$ and not at most $1$ as before (since not all the client edges are server edges, perhaps the best way to cover a client edge is to take a path of length 2 that covers it, the density of the corresponding star is $\frac{1}{2}$). Now $cost(e) \leq 2$ which changes slightly the constants in the analysis. When $v$ terminates, it adds an uncovered edge $e$ to the spanner only if $e$ is both a client and a server edge. These edges are the edges of $H_2$.

Note that since not all the edges are server edges, there may be client edges that cannot be covered by server edges, in which case there is no solution to the problem, and our algorithm covers only all the edges that may be covered by server edges. When we analyze the algorithm, we assume that there is a solution to the problem, otherwise $H^*$ is not defined. For other cases, we can restrict the client edges to be only edges that can be covered by server edges, and get a new problem that has an optimal solution $H^*$, and the approximation ratio we get is w.r.t to $H^*$.

For the analysis, there are slight differences as follows. First, we give costs only to client edges, since these are the only edges we need to cover. We give the costs as in the minimum 2-spanner algorithm. In particular, $cost(e)=1$ for $e \in H_2$.
Our goal is to show that
$$|H| \leq 8 \sum_{e \in C} cost(e) \leq O\left( \log{\frac{|C|}{|V(C)|}}\right) |H^*|.$$

The proof that $|H| \leq 8 \sum_{e \in C} cost(e)$ is exactly the same as the proof of Lemma \ref{cost}.
We next show that $\sum_{e \in C} cost(e) \leq O\left( \log{\frac{|C|}{|V(C)|}}\right) |H^*|.$
Let $r=\frac{|C|}{|V(C)|}$, and $f=\lceil \log{r} \rceil$. We define the sets $E_j$ according to the new definition of $r$. Let $E_1 = \{e \in C \setminus H_2: 0 < cost(e) \leq \frac{2}{r} \}$, and for $2 \leq j \leq f+1$, let $E_j = \{e \in C \setminus H_2: \frac{2^{j-1}}{r} < cost(e) \leq \frac{2^j}{r} \}.$ We define $E_0$ as before, and $E_{f+2} = H_2$. Since the stars added to $H_1$ in the algorithm have density at least $\frac{1}{2}$, then $cost(e) \leq 2$ for each edge $e \in C$. This gives, $C = \cup_{j=0}^{f+2} E_j$.
We next show the following.

\begin{lemma} \label{cs}
For every $0 \leq j \leq f+2$, $\sum_{e \in E_j} cost(e) = O(|H^*|).$
\end{lemma}

\begin{proof}
For $j=0$, $2 \leq j \leq f+1$, and $j=f+2$ the proof follows the cases $j=0$, $2 \leq j \leq f$ and $j=f+1$ in the proof of Lemma \ref{Hci}.
For $j=1$, it holds that
$\sum_{e \in E_1} cost(e) \leq \frac{2}{r} \cdot |C| \leq 2|V(C)| = O(|H^*|).$ In the first inequality, we use the fact that we give costs only to edges of $C$. The last equality follows from the fact that $H^*$ includes at least $\frac{|V(C)|}{4}$ edges, which we prove next.

Let $G_C=(V(C),C)$, and let $C_1,...,C_{\ell}$ be the connected components of $G_C$. Note that each connected component of $G_C$ includes at least two vertices (since it includes at least one edge of $C$), which means that the number $\ell$ of connected components is at most $\frac{|V(C)|}{2}$. For a connected component $C_i$, denote by $n_i$ the number of vertices in $C_i$, so that $|V(C)| = \sum_{i=1}^{\ell} n_i$. For a connected component $C_i$, denote by $H_i$ all the edges of $H^*$ that cover the edges in $C_i$. It holds that $|H_i| \geq n_i-1$ since $C_i$ is connected, and the edges of $H_i$ need to connect all the vertices in $C_i$, otherwise there is an edge in $C_i$ which is not covered in $H^*$. In addition, for each edge $e \in H_i$, at least one of the vertices of $e$ is in $C_i$, otherwise it cannot cover an edge in $C_i$. It follows that an edge $e \in H^*$ can be in at most two different subsets $H_i,H_j$. This gives
$$|H^*| \geq \frac{1}{2}\sum_{i=1}^{\ell}|H_i| \geq \frac{1}{2} \sum_{i=1}^{\ell} (n_i - 1) = \frac{1}{2}\sum_{i=1}^{\ell} n_i-\frac{1}{2}\ell \geq \frac{|V(C)|}{2} - \frac{|V(C)|}{4} = \frac{|V(C)|}{4},$$
which completes the proof.
\end{proof}

By Lemma \ref{cs}, we get
$\sum_{e \in C} cost(e) = \sum_{j=0}^{f+2} \sum_{e \in E_j} cost(e) = O(\log{r})|H^*|.$
Since $|H| \leq 8 \cdot \sum_{e \in C} cost(e)$, we have
$|H| \leq 8 \sum_{e \in C} cost(e) \leq O(\log{r}) |H^*|$, which shows an approximation ratio of $O(\log{\frac{|C|}{|V(C)|}})$.

In addition, we can show that $\sum_{e \in C} cost(e) \leq O(\log{\Delta_S})$, following the proof of Lemma \ref{weighted}, by replacing $w(S)$ and $w(H)$ by $|S|$ and $|H|$. This shows an approximation ratio of  $O(\log{\Delta_S})$ to the problem.

Note that in the minimum 2-spanner problem, $\frac{m}{n}$ is half of the average degree in $G$, and $\Delta$ is the maximum degree in $G$, hence an approximation ratio of $O(\log{\frac{m}{n}})$ is better than $O(\log{\Delta})$. However, in the client-server variant, it may be the case that $\Delta_S \leq \frac{|C|}{|V(C)|}$ depending on the client and server edges in $G$.
The time analysis is the same as in the minimum 2-spanner problem. Note that there may be at most $O(\log{\Delta_S})$ different values for $\rho$ because we consider only stars composed of server edges.
This completes the proof of Theorem \ref{csThm}.

\section{Distributed approximation for MDS} \label{sec:MDS}

In this section, we show that our algorithm can be modified to give an efficient algorithm for the minimum dominating set (MDS) problem, \emph{guaranteeing} an approximation ratio of $O(\log{\Delta})$.
In the MDS problem the goal is to find a minimum set of vertices $D$ such that each vertex is either in $D$ or has a neighbor in $D$.
Our algorithm for MDS has the same structure of the algorithm of Jia et al. \cite{jia2002efficient}, but it differs from it in the mechanism for symmetry breaking. Our approach guarantees an approximation ratio of $O(\log{\Delta})$, where in \cite{jia2002efficient} the $O(\log{\Delta})$-approximation ratio holds only in expectation. The following states our results for MDS.

\begin{theorem}
\label{theorem:MDS}
There is a distributed algorithm for the minimum dominating set problem in the \congest model that guarantees an approximation ratio of $O(\log{\Delta})$, and takes $O(\log{n} \log{\Delta})$ rounds w.h.p.
\end{theorem}

For MDS, we define the star $S_v$ centered at the vertex $v$ as the set of \emph{vertices} that contains $v$ and all of its neighbors. Note that there is only one star centered at each vertex, which simplifies both the algorithm and its analysis. The density of a star $S$ with respect to a subset of vertices $U$, denoted by $\rho(S,U)$, is defined as $|S \cap U|$. The density of a vertex $v$ with respect to $U$, denoted by $\rho(v,U)$, is defined as $|S_v \cap U|$. The definition of the rounded density is the same as for our algorithm for the minimum 2-spanner problem. 

A vertex $v$ maintains a set $U_v$ that contains all the vertices in $S_v$ that are still not covered by the vertices that have already been added to the dominating set, where a vertex is covered by a set if it is in that set or has a neighbor in that set. Our algorithm proceeds in iterations, where in each iteration the following is computed:
\\

\noindent\fbox{%
    \parbox{\textwidth}{%
    \vspace{-0.3cm}
\begin{enumerate}[rightmargin=0.2cm]
\item { Each vertex $v$ computes its rounded density $\rho_v = \tilde{\rho}(v, U_v)$, and sends it to its 2-neighborhood.}
\item { Each vertex $v$ such that $\rho_v \geq \rho_u$ for each $u$ in its 2-neighborhood is a \textit{candidate}. Vertex $v$ informs its neighbors that it is a candidate. Let $C_v = S_v \cap U_v$.}
\item { Each candidate $v$ chooses a random number $r_v \in \{1,...,n^4\}$ and sends it to its neighbors.}
\item { Each uncovered vertex that is covered by at least one of the candidates, votes for the first candidate that covers it according to the order of the values $r_v$. If there is more than one candidate with the same minimum value, it votes for the one with the minimum ID.}
\item { If $v$ receives at least $\frac{|C_v|}{8}$ votes from vertices it covers then it is added to the dominating set.}
\item { Each vertex updates the set $U_v$ by removing from it vertices that are now covered. If $U_v = \emptyset$, $v$ outputs 1 if and only if it was added to the dominating set in the previous step.}
\end{enumerate}
    \vspace{-0.4cm}
}}

~\\

A crucial difference from our spanner approximation algorithm is that the densities are now based on the number of uncovered neighbors of $v$, and not the number of uncovered edges that can be potentially covered by a star. For this reason, all the computations in the algorithm can be implemented efficiently in the \congest model.

The analysis of our MDS algorithm follows the same lines as the analysis of our minimum 2-spanner algorithm. We denote by $D$ the dominating set produced by the algorithm, and by $D^*$ a minimum dominating set. We assign each vertex $v$ with a value $cost(v)$,
which equals $\frac{1}{\rho}$ if $v$ is covered for the first time by a candidate having density $\rho$ that $v$ votes for, and otherwise, $cost(v) = 0$.
We show that $|D| \leq 8 \sum_{v \in V} cost(v) \leq O(\log{\Delta}) |D^*|$, which implies our claimed approximation ratio.

\begin{lemma}
\label{lemma:MDS}
$|D| \leq 8 \cdot \sum_{u \in V} cost(u)$.
\end{lemma}

\begin{proof}
The proof is similar to the proof of Lemma \ref{cost}. For a vertex $v \in D$ we denote by $Votes(v)$ the vertices that vote for $v$. If $v$ is added to $D$ then it holds that at least $\frac{|C_v|}{8}$ vertices vote for it. The cost of each of these vertices is $\frac{1}{\rho}$, where $\rho$ is the density of $v$, which is $|C_v|$, by definition. Hence, for each vertex $v \in D$, it holds that $\sum_{u \in Votes(v)} cost(u) \geq \frac{1}{\rho} \cdot \frac{|C_v|}{8}= \frac{1}{8}.$ Since each vertex $u$ is in at most one set $Votes(v)$, summing over all the vertices in $D$ gives that $|D| \leq 8 \cdot \sum_{v \in D} \sum_{u \in Votes(v)} cost(u) \leq  8 \cdot \sum_{u \in V} cost(u).$
\end{proof}

The proof that $\sum_{v \in V} cost(v) \leq O(\log{\Delta}) |D^*|$ is similar to the proof of Lemma \ref{weighted}, where $Stars^*$ is replaced by $D^*$, edges are replaced by vertices, and $w(S)$ is replaced by $1$ (note that the equality $\sum_{S \in Stars^*} w(S) = 2w(H^*)$ is replaced by $\sum_{v \in D^*} 1 = |D^*|$).  Together with Lemma~\ref{lemma:MDS}, this proves the approximation ratio of $O(\log\Delta)$.

For the time analysis, the main difference is that for each vertex $v$ there is only one star $S_v$, which simplifies the proof (Claim \ref{iteration_p} is no longer required).
Let $\rho=\max_{v \in V}{\rho_v}$ at the beginning of iteration $i$. We define the potential function $\phi = \sum_{v:\rho_v=\rho}|C_v|$.
If the value of $\rho$ does not change between iterations, the value of $\phi$ can only decrease between iterations. By the definition of the densities, the density of a vertex $v$ is a at most $|S_v|$ which is at most $\Delta + 1$. Since we round the densities there may be at most $O(\log{\Delta})$ different values for $\rho$. Following the same analysis as in the analysis of our minimum 2-spanner algorithm (with the difference that edges are replaced by vertices, and a candidate is a vertex and not a star) we can show that if the value of $\rho$ does not change between iterations, then the potential function $\phi$ decreases by a multiplicative factor between iterations in expectation. This gives a time complexity of $O(\log{n}\log{\Delta})$ rounds w.h.p. Together with the approximation ratio, this proves Theorem~\ref{theorem:MDS}.

\section{Distributed $(1+\epsilon)$-approximation for spanner problems} \label{sec:epsilon}

In this section, we show distributed $(1+\epsilon)$-approximation algorithms for spanner problems, following the framework of a recent algorithm for covering problems \cite{ghaffari2017complexity} (see Section 7).\footnote{The presentation of the framework in~\cite{ghaffari2017complexity} is slightly different and goes through an intermediate \emph{SLOCAL} model.}
In a nutshell, the vertices invoke a network decomposition algorithm on the graph $G^r$, for a value of $r=O(\log{n}/\epsilon)$ that can be computed by all vertices locally, given $\epsilon$ and a polynomial bound on $n$. This decomposes the graph into clusters of logarithmic diameter, colored by a logarithmic number of colors. Finally, by increasing order of colors, the vertices of each color select edges for the spanner. We show that indeed clusters of the same color can make their choices in parallel, and that the method of choosing edges to the spanner results in a $(1+\epsilon)$ approximation factor, giving the following.

\epsilonAlg*

\begin{proof}
We start by describing a sequential $(1+\epsilon)$-approximation algorithm, and then explain how to implement it in the \local model using network decomposition.
In the algorithm, the vertices start adding edges to the spanner $H$, which is initialized to be empty, while keeping track of all the edges covered by edges of $H$. At the beginning, all the edges are uncovered. To describe how this is done, we need the following notation. For a given integer $d$, denote by $B_d(v)$ the subgraph of all the vertices within distance at most $d$ from $v$ and all the edges between them. For a vertex $v$ and $d \geq 1$, let $g(v,d)$ be the size of an optimal spanner for all of the uncovered edges in $B_d(v)$ (notice that the spanner can use both covered or uncovered edges of the whole graph $G$). 

We process the vertices according to a given order $v_1,v_2,...,v_n$. In step $i$, we look for the smallest radius $r_i$ such that $g(v_i,r_i+2k) \leq (1+\epsilon)g(v_i,r_i)$. Since an optimal spanner has size at most $n^2$, increasing the radius without the condition being met can only happen at most $r_i = O(\log{n}/\epsilon)$ times. We add to $H$ an optimal spanner for all the uncovered edges in $B_{r_i+2k}(v_i)$, and mark all the edges covered by the new edges of $H$ as covered. In particular, all the edges of $B_{r_i+2k}(v_i)$ are covered after this step. Note that an optimal spanner for $B_{r_i+2k}(v_i)$ is contained in $B_{r_i+3k}(v_i)$, which shows that step $i$ depends only on a polylogarithmic neighborhood around $v_i$. 

We next prove the approximation ratio of the algorithm.
Denote by $E_i$ all the edges of $B_{r_i}(v_i)$ that are uncovered before step $i$. Since all the edges of $B_{r_i+2k}(v_i)$ are covered after step $i$, it follows that $E_i$ and $E_j$ are at distance at least $2k+1$ for $i \neq j$. Let $H^*$ be an optimal spanner, and let $H_i^*$ be the minimum set of edges in $H^*$ that covers $E_i$. 
By the definition of a $k$-spanner, $H_i^*$ is contained in $B_{r_i+k}(v_i)$, which shows that the subsets $H_i^*$ are disjoint. In step $i$, we added to $H$ at most $(1+\epsilon)g(v_i,r_i) \leq (1+\epsilon) |H_i^*|$ edges, where the inequality follows since $g(v_i,r_i)$ is the size of an optimal spanner for $E_i$ where $H_i^*$ is a spanner for $E_i$. Since $\cup_{i=1}^{n} H_i^* \subseteq H^*$ and the subsets $H_i^*$ are disjoint, summing over all $i$ gives $|H| \leq (1+\epsilon)|H^*|$, which completes the approximation ratio proof.

We now show how to implement the algorithm in the \local model (see also Proposition 3.2 in \cite{ghaffari2017derandomizing}). Let $r = O(\log{n}/\epsilon)$ be such that $r > r_i+4k$ for all $i$, 
and consider the graph $G^r$ on the same set of vertices, where two vertices are connected if they are at distance at most $r$ in the network graph $G$. Notice that in the \local model, any algorithm on $G^r$ can be simulated by the vertices of $G$ with an overhead of $r$ rounds.
The vertices invoke the randomized network decomposition algorithm of Linial and Saks \cite{linial1993low} on the graph $G^r$. This algorithm decomposes a graph into clusters of diameter $O(\log{n})$ that are colored with $O(\log{n})$ colors, within $O(\log^{2}{n})$ rounds. Invoked on $G^r$, this completes in $poly(\log{n}/\epsilon)$ rounds. 

We assign a vertex $v$ the label $(q_v,ID_v)$ where $q_v$ is the color of the cluster of $v$ and $ID_v$ is the id of $v$. The lexicographic increasing order of the labels provides the order of the vertices.
The distributed $k$-spanner algorithm runs in $O(\log{n})$ phases, where in each phase $\ell$, the vertices of color $\ell$ are active, and collect all of the information of their cluster in $G^r$ and its neighbors. Since the diameter of each cluster is at most $O(\log{n})$, this completes in $poly(\log{n}/\epsilon)$ rounds. Each vertex of the cluster then locally simulates the sequential algorithm for all the vertices in its cluster, according to their order. It can do so, since the sequential algorithm depends only on $r$-neighborhoods of vertices, and every two vertices in the same $r$-neighborhood are neighbors in $G^r$, which means they are either in the same cluster or in two clusters with different colors. This guarantees that the algorithm can indeed be executed in parallel for vertices of the same color. This completes the proof.
\end{proof}

The correctness of the algorithm relies only on the fact that the definition of $k$-spanners is local: an optimal spanner for $B_d(v)$ is contained in $B_{d+k}(v)$. Hence, the algorithm can be adapted similarly to the weighted, directed and client-server variants. In the weighted case the complexity is $O(poly(\log{(nW)}/\epsilon))$, where $W$ is the ratio between the maximum and minimum positive weights of an edge.

\paragraph*{Acknowledgment:}
We would like to thank Seri Khoury for fruitful discussions.

\bibliography{spanner}
\bibliographystyle{plain}

\end{document}